\documentclass[11pt]{article}
\usepackage[T1]{fontenc}
\usepackage[utf8]{inputenc}
\usepackage[margin=1in]{geometry}
\usepackage{graphicx}
\usepackage{mathtools}

\usepackage[english]{babel}
\usepackage{amsmath}
\usepackage{amsfonts}
\DeclareMathAlphabet{\mathpzc}{OT1}{pzc}{m}{it}
\newcommand{\pzp}{\mathpzc{p}}
\newcommand{\pzq}{\mathpzc{q}}
\usepackage[parfill]{parskip}
\usepackage{braket}
\newcommand{\sr}{r}

\newcommand{\bg}{\mathbf{g}}
\newcommand{\DD}{\mbox{\small\ensuremath{D}}}
\newcommand{\sign}{\mathrm{sign}}
\newcommand{\hatp}{\widehat{\mathbf{p}}^{1}}

\newcommand{\Dp}[1]{\prescript{#1}{}{\bbD_\bgamma}}
\newcommand{\Dq}[1]{\prescript{#1}{}{\bbD^*_\bgamma}}
\newcommand{\bp}{\mathbf{p}}
\newcommand{\matL}{\mathcal{L}}
\newcommand{\matD}{\mathcal{D}}
\newcommand{\matQ}{\mathcal{Q}}
\newcommand{\matB}{\mathcal{B}}
\newcommand{\matS}{\mathcal{S}}
\newcommand{\matT}{\mathcal{T}}
\newcommand{\matZ}{\mathcal{Z}}
\newcommand{\matF}{\mathcal{F}}
\newcommand{\tmatT}{\widetilde{\mathcal{T}}}

\newcommand{\lp}{\left(}
\newcommand{\rp}{\right)}

\newcommand{\bphi}{\boldsymbol{\varphi}}

\newcommand{\bpsi}{\boldsymbol{\varPsi}}

\newcommand{\bbD}{\mathbb{D}}

\newcommand{\bbV}{\mathbb{V}}

\newcommand{\bbW}{\mathbb{W}}
\newcommand{\bbN}{\mathbb{N}}
\newcommand{\bbL}{\mathbb{L}}

\newcommand{\bbC}{\mathbb{C}}
\newcommand{\bbR}{\mathbb{R}}

\usepackage{upgreek}
\newcommand{\linear}[2]{\mathfrak{L}\lp{#1, #2}\rp}
\newcommand{\hol}[2]{\mathfrak{Hol}\lp{#1, #2}\rp}
\usepackage{amsmath}

\usepackage{color,colortbl}
\usepackage{xcolor}
\definecolor{darkgreen}{RGB}{0,170,0}
\definecolor{darkpink}{RGB}{255,150,180}
\definecolor{lucapink}{rgb}{0.9,0.3,0.9}
\definecolor{applegreen}{rgb}{0.55, 0.71, 0.0}
\newif\ifshowcorrections \showcorrectionsfalse   % <-- false: \cf/\cb/\mmtext text kept but BLACK
\newif\ifshownotes       \shownotesfalse         % <-- false: drop every \luca/\lfnote/\mmnote/\mmdel
\newif\ifrevisionmarks   \revisionmarkstrue
\ifrevisionmarks\else \showcorrectionsfalse \shownotesfalse \fi
\ifshowcorrections
  \definecolor{claudered}{rgb}{0.85,0,0}
  \definecolor{claudeblue}{rgb}{0,0.25,0.85}
  \newcommand{\cf}[1]{{\color{claudered}{#1}}}
  \newcommand{\cb}[1]{{\color{claudeblue}{#1}}}
  \newcommand{\mmtext}[1]{{\color{applegreen}{#1}}}
\else
  \definecolor{claudered}{rgb}{0,0,0}
  \definecolor{claudeblue}{rgb}{0,0,0}
  \newcommand{\cf}[1]{#1}
  \newcommand{\cb}[1]{#1}
  \newcommand{\mmtext}[1]{#1}
\fi
\ifshownotes
  \newcommand{\luca}[1]{{\color{lucapink}{#1}}}
  \newcommand{\lfnote}[1]{{\color{lucapink}[\textbf{LF}: {#1}]}}
  \newcommand{\mmnote}[1]{{\color{applegreen}[\textbf{MM}: {#1}]}}
  \newcommand{\mmdel}[1]{{\color{applegreen}{\st{#1}}}}
\else
  \newcommand{\luca}[1]{}
  \newcommand{\lfnote}[1]{}
  \newcommand{\mmnote}[1]{}
  \newcommand{\mmdel}[1]{}
\fi
\newcommand{\jump}[1]{[\mkern-3mu[#1]\mkern-3mu]}% jump across a discontinuity (Q at the reset)

\usepackage{amssymb}
\usepackage{amsthm}
\numberwithin{equation}{section}
\theoremstyle{plain}
\newtheorem{theorem}{Theorem}[section]
\newtheorem{lemma}[theorem]{Lemma}

\newtheorem{proposition}[theorem]{Proposition}

\theoremstyle{definition}
\newtheorem{definition}[theorem]{Definition}
\theoremstyle{remark}
\newtheorem{remark}[theorem]{Remark}

\newcommand{\bgamma}{{\boldsymbol{\upgamma}}}

\usepackage[hidelinks]{hyperref}
\hypersetup{pdftitle={Spectral theory for population density dynamics of spiking neurons with refractoriness},
  pdfauthor={L. Falorsi, G. V. Vinci, and M. Mattia}}

\title{Spectral theory for population density dynamics of spiking neurons with refractoriness%
\thanks{Work partially funded by the NextGenerationEU and MUR (PNRR-M4C2I1.3) project MNESYS (PE0000006--DD 1553 11.10.2022) and project EBRAINS-Italy (IR0000011--DD 101 16.6.2022) to MM.}}
\author{%
  Luca Falorsi\thanks{Dept.\ of Mathematics, ``Sapienza'' University, 00185 Rome, Italy (\texttt{luca.falorsi@gmail.com}).}
  \and Gianni V.\ Vinci\thanks{Dept.\ of Neuroscience, Istituto Superiore di Sanit\`a, 00161 Rome, Italy.}
  \and Maurizio Mattia\thanks{Dept.\ of Neuroscience, Istituto Superiore di Sanit\`a, 00161 Rome, Italy (\texttt{maurizio.mattia@iss.it}).}%
}
\date{}
\graphicspath{{Figures/}}

\begin{document}

\maketitle

% REQUIRED
\begin{abstract}
\mmtext{Incorporating an absolute refractory period into the population density approach for spiking neurons remains an open problem, despite evidence that refractoriness can strongly affect nonlinear transfer functions and network stability. 
We develop a rigorous operator-theoretic framework for neuronal population dynamics with a finite refractory time by augmenting the state space to include refractory history and formulating the problem as a non-self-adjoint boundary eigenvalue problem for the Fokker-Planck operator. 
This yields a complete spectral characterization of the generator, proves dissipativity and the existence of a contraction semigroup, and identifies defective eigenvalues as exceptional points where oscillatory modes emerge from coalescing relaxational modes. 
Within the framework of linear response theory, we also derive an exact transfer function that accounts for boundary conditions modulated by external input, correcting previous heuristic derivations and revealing additional threshold-noise contributions.
Using this transfer function under a mean-field approximation, we further show that refractoriness in populations of interacting neurons can facilitate the onset of limit cycles, that is, stable oscillations in the firing rate.
These results provide a rigorous foundation for spectral decomposition methods in computational neuroscience, opening the way to their further rigorous mathematical analysis.}
\end{abstract}

\section{Introduction}

%\subsection{Population density approach}
Population density dynamics describes the evolution of the probability density of neurons as a function of their membrane potential, capturing the collective behavior of a homogeneous ensemble of spiking neurons in the thermodynamic limit \cite{Knight1996,brunel1999fast,knight2000dynamics,Mattia2002}. 
The underlying neuron model is given by integrate-and-fire dynamics \cite{tuckwell1988introduction,Gerstner2014} under diffusion approximation.
Specifically, each neuron follows an i.i.d. diffusion process given $dV_t/\tau_m = \left[A(V_t) +\mu(t)\right]dt + \sqrt{2\DD(t)/\tau_m}dB_t$, $V_t\in(\alpha,\theta)$, where $B_t$ denotes a Wiener process. 
When the membrane potential reaches the threshold $\theta$, the neuron emits a spike (i.e., an action potential). 
The potential is reset to a \mmtext{post-spike} lower value $H < \theta$ and resumes the original subthreshold dynamics after an absolute refractory period $\tau_0 \ge 0$. Here we measure time as a multiple of the membrane time constant ($\tau_m > 0$) by assuming $\tau_m = 1$. A reflecting barrier is set at the minimum potential value \mmtext{$\alpha \leq H$.}   %TODO Specify better We will either assume $\alpha<H$ to be a finite value, or a confining force field $f$.
Here $\bgamma(t) := \{\mu(t), \DD(t)\}$ represents the infinitesimal moments of the current that the neuron receives and $A: C([\alpha, \theta], \bbR)$ is a continuous function that gives the membrane leakage.
The collective behaviour of the network composed of such neurons is described by the Fokker-Planck (FP) equation 
\begin{align}\label{eq:FokkerPlanck}
    \mmtext{\partial_t p_t(v) = \mathcal{L}_{\bgamma(t)}p_t(v) = -\partial_v\mathcal{S}_{\bgamma(t)}{p_t}(v), \,(t,v) \in [0,\infty)\!\times\! (\alpha, \theta)}
\end{align}
which represents the evolution of the probability density $p_t(v)$ of the membrane potential of the neurons over time. 
Here $\mathcal{S}_{\bgamma}{p}(v):= \left[A(v)+\mu\right]p(v) - \DD\partial_v p(v)$ denotes the probability current (i.e, the flux of realizations), which at the threshold potential gives the fraction of neurons emitting an action potential, i.e., the firing rate of the population:
\begin{align}\label{eq:FiringRate}
   \mmtext{\nu(t):= \mathcal{S}_{\bgamma(t)}{p_t}(\theta).}
\end{align}
% \mmnote{FP eq. and firing rates should be numbered equations.}
 
To correctly model the physics of the system, absorbing boundary conditions at $\theta$ and {\it reinjection boundary conditions} are imposed, such that all the probability mass escaping at $v=\theta$ is reinjected at the reset value $v=H$:
\begin{align}\label{eq:abs-reinjbc}
   p_t(\theta)=0, \qquad \mathcal{S}_{\bgamma(t)}{p_{t}}(H^+) - \mathcal{S}_{\bgamma(t)}{p_{t}}(H^-)=\mathcal{S}_{\bgamma(t)}{p_{t-\tau_0}}(\theta) = \nu(t-\tau_0)
\end{align}
We further impose continuity of the probability density at the reset\footnote{Given a function $f$ right (resp. left) continuous at $H$ we denote with $f(H^+)$ (resp. $f(H^-)$) its right (resp. left) limit at $H$.} $H$ and a reflecting barrier at $\alpha$:
\begin{align}
    p_t(H^+) = p_t(H^-), \qquad \mathcal{S}_{\bgamma(t)}{p_t}(\alpha) = 0
\end{align}
The FP equation introduced above, widely used in the physics and applied mathematics literature, has been recently derived rigorously from the microscopic dynamics of integrate-and-fire models \cite{liu2022rigorous}.
When considering coupling among neurons, the mean-field approximation is employed, making the moments $\mu(t), \DD(t)$ depend self-consistently on the past values of the firing rate, making the equation nonlinear \cite{Amit1997}. 
This gives rise to interesting and rich behaviour for the system that can display multistability \cite{Brinkman2022}, limit cycles \cite{brunel1999fast} and finite time blowups \cite{caceres2011analysis}.

\subsection{Spectral decomposition of the dynamics}

\mmtext{In the framework of this \textit{population density approach},} the analytical characterization of population dynamics is mostly focused on the steady states \mmtext{\cite{Amit1997, omurtag2000dynamics},} and on perturbative expansions via linear response theory \mmtext{\cite{Treves1993, brunel1999fast, brunel2000dynamics, Lindner2001, Richardson2007}.}
To overcome these limitations, a spectral decomposition method was introduced in \cite{Abbott1993, Knight1996, Mattia2002}. %valid in the absence of the absolute refractory period $\tau_0=0$. 
The population density is projected onto a time-dependent basis formed by the eigenfunctions of the FP operator: $p_t = \sum_{n\in \mathbb{N}} a_n(t) \varphi_n(\cdot;  \bgamma(t))$, where $a_n(t):= \langle \psi_n(\cdot; \bgamma(t)), p_t\rangle$ and 
%$\{\varphi_n(\cdot; \bgamma)\}_n$, $\{\psi_n(\cdot; \bgamma)\}_n$ 
\mmtext{$\varphi_n(\cdot; \bgamma)$, $\psi_n(\cdot; \bgamma)$}
represent respectively the eigenfunctions of $\mathcal{L}_{\bgamma}$ and its adjoint, forming a biorthonormal system.
Interestingly, the spectrum displays a hierarchy of timescales such that only the slowest modes contribute to the firing rate dynamics 
\mmtext{\cite{Mattia2002, Augustin2017}}.
Despite its effectiveness, the spectral properties of $\mathcal{L}_\gamma$ were never fully characterized in a rigorous functional analytic setting. 
% Here we assume $\mathbb{H} = L^2\left(\left(\alpha, H\right)\cup \left(H, \theta\right)\right)$ and  $\mathbb{D}$ the dense linear subspace of $H^2\left(\left(\alpha, H\right)\cup \left(H, \theta\right)\right)$ where the boundary conditions are satisfied. 
Most notably, a proof of the completeness of the biorthogonal system is still missing. The problem cannot be reduced to the study of a self-adjoint operator, and the spectral theorem cannot be applied due to the boundary conditions. 
Ultimately, this is a consequence of the broken detailed balance in the physical system described by the \mmtext{FP} equation \cite{Risken1984}.

Furthermore, incorporating in the population density approach the presence of an absolute refractory period $\tau_0$ has been so far an unsolved problem.
\mmtext{Indeed, the presence of a $\tau_0 > 0$} is often ignored, regarded as negligible \mmtext{\cite{Berry1998, Mattia2002}}. 
Despite this, the few studies that incorporate it \cite{fusi1999collective} \mmtext{show that the absolute refractory period is an important nonlinearity in current-to-rate transfer functions and a key contributor to stable equilibria.
Theoretical evidence also comes from alternative approaches such as Spike-Response Models \cite{srm} and the Refractory Density Method \cite{rdm}, where the presence of $\tau_0$ is explicitly taken into account in establishing the network dynamics.}

\mmtext{In this work, we extend the spectral decomposition of neuronal dynamics to include an absolute refractory period $\tau_0$.
We do so by introducing an augmented state space adding a domain $(0,\tau_0)$ where realizations enter after crossing the threshold $\theta$ before being reinjected in $H$. 
This leads to a non-self-adjoint boundary eigenvalue problem on the union of three domains $(\alpha, H)\cup(H, \theta)\cup(0,\tau_0)$. 
in Section \ref{sec:spectrum} we analyze how $\tau_0$ modifies the spectrum and the system's dynamics, derive the corresponding characteristic equation for eigenvalues and eigenfunctions, and generalize the known result for $\tau_0 = 0$.
Finally, in Section \ref{sec:tf}, we study the response to small time-dependent inputs through a perturbation around the stationary solution. 
The resulting frequency-domain transfer function, obtained from the resolvent, yields analytic expressions for the coupling coefficients between stationary and non-stationary modes and clarifies the role of refractoriness in stability and in the transition from stationary to periodic dynamics.}

\section{Augmenting the state space}
\label{sec:operator}

The main obstacle in extending the spectral decomposition to incorporate the presence of an absolute refractory period $\tau_0 > 0$ is the boundary condition (BC) \eqref{eq:abs-reinjbc} about the flux reinjection, which becomes dependent on the past firing rate values. This results in a non-homogenous domain for the \mmtext{FP} operator.  
In other words, at a given time $t$, knowing the density $p_t$ is not sufficient to determine the system's evolution in the future. To recover the Markovianity of the system we need to keep track of all the realizations of neurons that are in a refractory state. 
We can therefore augment the state space with a new density $p^\sr_t(\tau)$ that stores the firing rate in the past:
\begin{align}
    p^\sr_t(\tau) := \nu(t - \tau)\quad \forall \tau\in (0,\tau_0)
\end{align}
The evolution of this \mmtext{density} follows a linear transport equation:
\begin{align}
    \frac{\partial}{\partial t} p^\sr_t(\tau) + \frac{\partial}{\partial \tau} p^\sr_t(\tau) = 0
\end{align}
% We can then define $\varrho(t) := (p(\cdot,t), r(\cdot, t))$ and $\matT^t = \matT_{\mu(t),\sigma(t)}:= \lp \matL_{\mu(t), \sigma(t)}, -\partial_\tau \rp$. 
The evolution of the system is therefore given by the two PDEs:
\begin{align}\label{eq:pde-tau}
\begin{cases}
\partial_t p_t(v) &= -\partial_v\left[\lp A(v) + \mu(t)\rp p_t(v)\right] + {\DD(t)}\partial_v^2p_t(v)\qquad  v\in (\alpha, H)\cup(H,\theta)\\
    \partial_t p^\sr_t(\tau) &= -\partial_\tau p^\sr_t(\tau,t)  \qquad \tau \in (0,\tau_0)
\end{cases} \, ,
\end{align}
% As Banach space we will use the Hilbert space $ \bbH_{v,\tau}= L^2_{w}\lp(-\infty, v_t)\rp \times L^2\lp(0, \tau_0)\rp$ with scalar product:

% \begin{align}
%     \langle (f,r), (g,q)\rangle = \int_{-\infty}^{v_t} w(v)\overline{f(v)}g(v)dv + \int_{0}^{\tau_0} \overline{r(\tau)}q(\tau)d\tau
% \end{align}

% \begin{align}
%     \tmatT_\upgamma^+: \bbD^+_{v,\tau}\subset \bbH_{v,\tau} &\to \bbH_{v, \tau}\\
%     (p(v), r(\tau)) &\mapsto ( \mathcal{L}_{\mu, \sigma} p(v), \partial_\tau r(\tau)) \qquad \upgamma = (\mu, \sigma)
% \end{align}
coupled through the boundary conditions which can be rewritten as:
\begin{align}\label{eq:BC-aug}
\begin{cases}
    \text{BC1: }p_t(\theta, t) = 0\\
    \text{BC2: }p_t(H^-) - p_t(H^+) = 0\\
    \text{BC3: }\mathcal{S}_{\bgamma(t)}p_t(H^+) -  \mathcal{S}_{\bgamma(t)}p_t(H^-) = p^\sr_t(\tau_0)\\
    \text{BC4: }\mathcal{S}_{\bgamma(t)}p_t(\theta) = p^\sr_t(0)  \\
    \text{BC5: }\mathcal{S}_{\bgamma(t)}p_t(\alpha)=0\\
\end{cases} \, .
\end{align}
These equations are now homogeneous in the augmented state space $(p_t, p^\sr_t)$. 
\subsection{Definition of the evolution operator}
The evolution of the system given by \eqref{eq:pde-tau} can be rewritten as an abstract differential equation:
\begin{align}\label{eq:abstract-ode}
    \dot{\mathbf{p}}_t = \matT_{\bgamma(t)}\mathbf{p}_t \, ,
\end{align} 
where $\matT_{\bgamma(t)}$ is a linear operator on a suitable functional space. 
We will work on $\mathrm{L}^\pzp$ spaces defined on the union of three compartments
\footnote{Since $L^\pzp\lp(\alpha, H)\rp \times L^\pzp\lp(H, \theta\rp$ and $L^\pzp\lp(\alpha,\theta\rp$ can be identified, we will alternatively use the notation $\mathbf{p}=(p^-, p^+, p^\sr)=(p, p^\sr)$ for the elements of $\bbL^\pzp$, where $p^-$ (resp. $p^+$) is the restriction of $p\in L^\pzp ((\alpha, \theta))$ to $L^\pzp((\alpha, H))$ (resp. to $L^\pzp((H, \theta))$).}
$(\alpha, H)\cup(H, \theta)\cup(0,\tau_0)$:
\begin{equation}
\begin{split}
    \bbL^\pzp &:= \{ \mathbf{p}=(p^{-}, p^{+}, p^{r}) \in \mathrm{L}^\pzp\lp\alpha, H\rp \times \mathrm{L}^\pzp\lp H, \theta\rp \times \mathrm{L}^\pzp\lp 0, \tau_0\rp\} \\
    &= \{\mathbf{p}=(p, p^{\sr}) \in \mathrm{L}^\pzp\lp\alpha,\theta\rp \times \mathrm{L}^\pzp\lp 0, \tau_0\rp\},\quad \pzp\in[1,\infty)
\end{split} 
\end{equation}
With duality pairing $\langle\cdot,\cdot\rangle:\bbL^\pzq\times\bbL^\pzp\to \bbC$ defined by:
\begin{equation}
\begin{split}
    \langle \mathbf{f},\mathbf{p}\rangle &:= \int_\alpha^\theta f(v)p(v)dv + \int_0^{\tau_0} f^\sr(\tau)p^\sr(\tau)d\tau \\
    &= \int_\alpha^{H^-} f^-(v)p^-(v)dv + \int_{H^+}^\theta f^+(v)p^+(v)dv +\int_0^{\tau_0} p^\sr(\tau)f^\sr(\tau)d\tau
\end{split}
\end{equation}
where $\pzp^{-1} + \pzq^{-1} = 1$ and $\mathbf{p}\in\bbL^\pzp, \mathbf{f}\in \bbL^\pzq$. We will assume $\pzp \in [1, \infty), \pzq\in (1, \infty]$
We further denote with $\bbW^\pzp$ the Sobolev space $\bbW^\pzp := \mathrm{W}^{2,\pzp}(\alpha,H)\times\mathrm{W}^{2,\pzp}(H,\theta)\times \mathrm{W}^{1,\pzp}(0,\tau_0)$. 
We can then define the evolution operator: 
\begin{definition}
   Let $\mu\in\bbR$, $\DD>0$ the evolution operator is the unbounded, linear operator 
    \begin{align}
        \matT_{\bgamma}:\Dp{\pzp} \subset \bbL^\pzp &\to \bbL^\pzp\\
        \mathbf{p}=(p, p^\sr) &\mapsto \lp \matL_\bgamma p, -\partial_\tau p^\sr\rp
    \end{align}
    with domain  $\Dp{\pzp}$ is given by the linear subspace of functions in the Sobolev space that satisfy the boundary conditions:
\begin{align}
    \Dp{\pzp} := \{\mathbf{p} \in \bbW^\pzp | \mathcal{B}_\bgamma \mathbf{p} = 0\}
\end{align}
where $\mathcal{B}_\bgamma: \bbW^\pzp \to \bbC^5$ is the boundary operator defined as:
\begin{align}\label{eq:boundary-op}
    \mathcal{B}_\bgamma \mathbf{p} = 
\begin{bmatrix}
    \matS_{\bgamma}{p^-}(\alpha)\\
    p^+(\theta)\\
    p^+(H) - p^-(H) \\
    \matS_{\bgamma}{p^+}(H) -  \matS_{\bgamma}{p^-}(H) - p^\sr(\tau_0)\\
    \matS_{\bgamma}{p^+}(\theta) - p^\sr(0)  \\
\end{bmatrix}    
= \begin{bmatrix}
    \lp A(v)+\gamma_1\rp p(\alpha) -\cf{\DD} p'(\alpha)\\
    p(\theta)\\
    p(H^+) - p(H^-) \\
    -{\DD}\lp p'(H^+) -  p'(H^-)\rp - p^\sr(\tau_0)\\
    -{\DD}p'(\theta) - p^\sr(0)  \\
\end{bmatrix}
\end{align}
\end{definition}
\begin{remark}
The operator $\matT_\bgamma$ is closed and densely defined $\forall \pzp \in [1, \infty)$.
\begin{proof}
To see that $\Dp{\pzp}\subset \bbL^\pzp$ is dense, it is sufficient to observe that the following subset is dense in $\bbL^{\pzp}$:
\vspace{-10pt}
\begin{equation*}
    \!\!\!\Set{\!\!(p, r)\in C^\infty(\alpha, \theta)\!\times\! C^\infty(0, \tau_0)| \!\!\begin{array}{l}
    \exists K_1, K_2 \text{ compact s.t. }\\
    supp(p)\!\subseteq \!K_1\!\subset\! (\alpha, \theta), \ supp(r)\!\subseteq \!K_2\!\subset\! (0, \tau_0)
\end{array} \!\!\!\!\!}\subset \Dp{\pzp}.\end{equation*}
Closeness of the operator follows from \cite[Remark 3.4.2]{mennicken2003non} 
\end{proof}
\end{remark}
We then define the ``firing rate'' operator:
\begin{align}
    N: \bbW^\pzp &\to \bbC\\
    \mathbf{p} &\mapsto N(\mathbf{p}) = p^\sr(0)
\end{align}
As a trace operator on Sobolev spaces, this operator is compact. Its restriction $N|_{\Dp{\pzp}}:\Dp{\pzp}\subset \bbL^\pzp \to \bbC$, considered as an unbounded operator on $\bbL^\pzp$ is relatively compact with respect to $\matT_\bgamma$. Notice that for $\mathbf{p} \in \Dp{\pzp}$ we have that the firing rate is given by the \mmtext{usually defined flux at the threshold $N(\mathbf{p}) = \mathcal{S}_\bgamma p(\theta)$ \cite{Ricciardi1988, Knight1996, brunel1999fast, Mattia2002}.}

We then turn our attention to the adjoint evolution operator.
\begin{theorem}

The adjoint of the evolution operator $\matT_\bgamma: \Dp{\pzp}\subset \bbL^\pzp \to \bbL^\pzp$ is the operator \begin{align}
    \matT^*_\bgamma: \Dq{\pzq}\subset \bbL^\pzq &\to \bbL^\pzq \\
    \mathbf{f}=(f,f^\sr) &\mapsto \lp \matL^*_\bgamma f, \partial_\tau f^\sr\rp
\end{align}
where $\matL^*_\bgamma f (v) = \cf{[A(v) + \mu]f'(v) + \DD f''(v)}$ is the backward Kolmogorov operator and the operator domain $\Dq{\pzq}$ is given by the adjoint boundary conditions:
\begin{align}
   \Dq{\pzq} := \{\mathbf{f} \in \bbW | \mathcal{B}^*\mathbf{f} = 0\}
\end{align}
where $\mathcal{B}^*: \bbW \to \bbC^5$ is the adjoint boundary operator defined as:
\begin{align}\label{eq:adj-bc}
    \mathcal{B}^*\mathbf{f} = \mathcal{B}^*(f, f^\sr) =
\begin{bmatrix}
    f'(\alpha)\\
    f'(H^+) - f'(H^-)\\
    f(H^+) - f(H^-) \\
    f(H) - f^\sr(\tau_0) \\
    f(\theta) - f^\sr(0)
\end{bmatrix}    
\end{align} 
\begin{proof}
    Let $\mathbf{f}\in \bbW^\pzq$, $\mathbf{p}=(p, p^r)\in \bbW^\pzp$, integrating by parts we have:
    \begin{align} \notag
        \langle \mathbf{f}, \matT_\bgamma \mathbf{p}\rangle =\ &\langle \matT^*_\bgamma \mathbf{f}, \mathbf{p}\rangle \\
        & \cf{-} \left[\cf{\DD\,}\partial_v f^-p^-\right]_\alpha^H \cf{-}  \left[\cf{\DD\,}\partial_v f^+p^+\right]_H^\theta -\left[ f^-S_\bgamma p^-\right]_\alpha^{H} -\left[ f^+S_\bgamma p^+\right]_H^{\theta} - [f^rp^r]_0^{\tau_0} \label{eq:surface-term}
    \end{align}
Then the surface term \eqref{eq:surface-term} vanishes $\forall\bp\in \Dp{\pzp}$ if and only if $\matB^* \mathbf{f}=0$. 
\end{proof}
\end{theorem}
Notice that the boundary conditions do not depend on $\bgamma$ and that for the 2$^\text{nd}$ and 3$^\text{rd}$ boundary condition we have $\matB^*\mathbf{f} = 0 \Rightarrow f\in \mathrm{W}^{2,\pzq}(\alpha, \theta)$.
Therefore, when considering the adjoint operator, it is more natural to work with only two compartments $(\alpha, \theta)\cup (0, \tau_0)$. 
We refer to Appendix~\ref{sec:supp-adjoint} for additional details.
Notice that while for $\pzq < \infty$ the domain $\Dq{\pzq}$ is dense in $\bbL^\pzq$, this is not true for $\pzq=\infty$.

\section{Spectrum of the evolution operator}
\label{sec:spectrum}

The spectrum of the evolution operator $\matT_\bgamma$ can be studied using standard tools from the theory of boundary eigenvalue problem \mmtext{\cite{mennicken2003non}}. 
% We will use \cite{mennicken2003non} as a reference. 
We start by defining the boundary eigenvalue operator function $\tmatT_\bgamma \in \mathfrak{hol}\lp\bbC, \mathfrak{L}\lp\bbW, \bbL^\pzp\rp\rp$ for our problem. This is defined as the holomorphic function that maps the complex parameter $\lambda\in\bbC$ to the operator: 
\begin{align}
    \tmatT_\bgamma(\lambda) = \lp\tmatT_\bgamma^\mathtt{D}(\lambda), \tmatT_\bgamma^\mathtt{R}(\lambda)\rp:\bbW^\pzp &\to \bbL^\pzp \times \bbC^5\\
    \mathbf{p}=\lp p, p^\sr\rp &\mapsto \lp \lp\matL_\bgamma p, -\partial_\tau p^\sr\rp - \lambda \mathbf{p}, \matB_\bgamma \mathbf{p}\rp \, ,
\end{align}
\mmtext{such that ${\bbD_\bgamma} = \mathrm{ker }\tmatT_\bgamma^\mathtt{R}(\lambda)$.}
We then define the fundamental matrix function as the holomorphic function $\matZ_\bgamma\in \hol{\bbC}{{\mathfrak{L}\lp\bbC^5, \bbW\rp}}$  \mmtext{parametrizing} the kernel of $\tmatT_\bgamma^\mathtt{D}$:
\begin{align}
    \matZ_\bgamma(\lambda): \bbC^5 &\to \bbW\\ \label{eq:fund-matrix}
    \lp a^-, b^-, a^+, b^+, a^\sr\rp &\mapsto 
    \begin{bmatrix}
        a^- f_1^-(\cdot, \lambda, \bgamma) + b^-f^-_2(\cdot, \lambda, \bgamma)\\
        a^+ f_1^+(\cdot, \lambda, \bgamma) + b^+f^+_2(\cdot, \lambda, \bgamma)\\
        a^\sr \exp(-\cdot\lambda)
    \end{bmatrix}
\end{align}
Where $f_1(\cdot, \lambda, \bgamma) , f_2(\cdot, \lambda, \bgamma) \in W^{2,\pzp}\lp\lp\alpha, \theta\rp\rp$ are two independent fundamental solutions of the equation $\matL_\bgamma f =\lambda f$ such that:
\begin{align}\label{eq:fund-sol}
    &\matL_\bgamma f_i(v, \lambda, \bgamma) = \lambda f_i(v, \lambda, \bgamma)\quad \forall v\in (\alpha, \theta)\\
    &\begin{bmatrix}
        f_1(\alpha, \lambda, \bgamma) &  f_2(\alpha, \lambda, \bgamma) \\ 
        S_1(\alpha, \lambda, \bgamma) &  S_2(\alpha, \lambda, \bgamma)
    \end{bmatrix} = Id_{\bbC^2}
\end{align}
Where we defined $S_i(\cdot, \lambda, \bgamma):= \matS_\bgamma f_i(\cdot, \lambda, \bgamma)$. 
It is straightforward to see that this boundary eigenvalue problem can be mapped to an equivalent five dimensional first-order system.
% (see supplementary material for additional details \luca{TODO refer to supplementary material}). 
The spectrum structure of $\tmatT_\bgamma$ and $\matT_\bgamma$ can be completely determined by studying the characteristic matrix $M_\bgamma\in\hol{\bbC}{\linear{\bbC^5}{\bbC^5}}$, defined as:
\begin{align}
    M_\bgamma(\lambda) = \matB_\bgamma \matZ_\bgamma(\lambda) \quad \forall \lambda \in \bbC
\end{align}

\begin{lemma}\label{thm:spectrum-bep}
    The couple $(\tmatT_\bgamma, \matZ_\bgamma)$ forms an abstract boundary eigenvalue problem \cite[p. 46-50]{mennicken2003non}. Such that $\tmatT_\bgamma$ is a Fredholm operator valued function and there exists $\mathcal{C}_\bgamma: \hol{\bbC}{\linear{\bbC^5\times \bbL^\pzp}{\bbL^\pzp \times \bbC^5}}$ and  $\mathcal{D}_\bgamma: \hol{\bbC}{\linear{\bbW^\pzp}{\bbC^5\times\bbL^\pzp}}$ such that $\mathcal{C}_\bgamma(\lambda)$, $ \mathcal{D}_\bgamma(\lambda)$ are invertible $\forall \lambda \in \bbC$ and the following decomposition holds:
    \begin{align}
        \tmatT_\bgamma(\lambda) = \mathcal{C}_\bgamma(\lambda) 
        \begin{bmatrix}
            M_\bgamma(\lambda) & \mathbf{0}\\
            \mathbf{0} & \mathbf{Id}_{\bbL^p}
        \end{bmatrix} \mathcal{D}_\bgamma(\lambda) \quad \forall \lambda\in \bbC
    \end{align}
    We then say that $\tmatT_\bgamma$ is globally equivalent on $\bbC$ to the canonical $\bbL^\pzp$ extension of the characteristic matrix $M_\bgamma$. 
    we then have:
    \begin{align}
        \sigma(\matT_\bgamma^*) = \sigma(\matT_\bgamma) = \sigma(\tmatT_\bgamma) = \sigma(M_\bgamma)=:\sigma(\bgamma)
    \end{align}
    Such that for each eigenvalue $\lambda_*$ of $\matT_\bgamma$ the numerical data (geometric
    multiplicity, algebraic multiplicity, partial multiplicities etc.)
    are the same as the numerical data of $\lambda_*$ as an eigenvalue of $M_\bgamma$.
    \begin{proof}
        This is a restatement of Theorem 1.11.1 and Lemma 1.11.2 in \cite{mennicken2003non} in the current context.
        \end{proof}
\end{lemma}
\begin{remark}\label{rm:var-par}
    Notice that $\tmatT_\bgamma^\mathtt{D}(\lambda): \bbW^\pzp \to \bbL^\pzp$ has a right inverse $\matQ_\bgamma(\lambda):\bbL^\pzp \to \bbW^\pzp$, holomorphic in $\lambda$, given by the variation of parameters formula: 
        \begin{align}\notag
            \bbW^\pzp \ni\mathbf{g} = \begin{bmatrix}
                g\\ g^\sr
            \end{bmatrix} \mapsto &\matQ_\bgamma(\lambda) \mathbf{g} = \\ &=\begin{bmatrix}
                v \mapsto\! \int_\alpha^vg(u)\dfrac{f_1(v, \bgamma, \lambda)f_2(u, \bgamma, \lambda) -  f_2(v, \bgamma, \lambda)f_1(u, \bgamma, \lambda)}{w_\bgamma(u)} du\quad \vspace{0.5em}\\ 
                \tau \mapsto\! \cf{-}e^{-\tau\lambda}\int_0^\tau e^{z\lambda}g^\sr(z)dz
            \end{bmatrix}
        \end{align}
        where $w_\bgamma : [\alpha, \theta]\to \bbR_+$ is the (rescaled) Wronskian for the functions $f_1$ and $f_2$ (see \cite[Section III.B]{vinci2024rosetta}): 
        \begin{align} \notag
            w_\bgamma(v) &= \DD[f_2(v, \bgamma, \lambda)f_1'(v, \bgamma, \lambda) - f_2'(v, \bgamma, \lambda)f_1(v, \bgamma, \lambda)] =\\ \label{eq:wronskian}
            &=f_1(v, \bgamma, \lambda)\matS_\bgamma f_2(v, \bgamma, \lambda) - f_2(v, \bgamma, \lambda)\matS_\bgamma f_1(v, \bgamma, \lambda) = \\
            & = \exp \lp  \int_\alpha^v \frac{A(u) + \mu}{\DD}\, du \rp > 0 \quad \forall v\in [\alpha, \theta] \notag
        \end{align}
\end{remark}
We can then write the characteristic matrix for our problem using Equations \eqref{eq:boundary-op} and \eqref{eq:fund-matrix}:
\begin{align}
     M_\bgamma(\lambda) = 
     \begin{bmatrix}
    0 & 1 & 0 & 0 & 0 \\ 
    0 & 0 & f_1(\theta, \lambda, \bgamma) & f_2(\theta, \lambda, \bgamma) & 0 \\
    -f_1(H, \lambda, \bgamma) & -f_2(H, \lambda, \bgamma) &  f_1(H, \lambda, \bgamma) & f_2(H, \lambda, \bgamma)& 0 \\ 
    -S_1(H, \lambda, \bgamma) & -S_2(H, \lambda, \bgamma) &  S_1(H, \lambda, \bgamma) & S_2(H, \lambda, \bgamma)& - \exp(\cf{-}\lambda\tau_0) \\
    0 & 0 & S_1(\theta, \lambda, \bgamma) & S_2(\theta, \lambda, \bgamma) & -1
     \end{bmatrix}
\end{align}

\begin{theorem}\label{thm:spectrum-prop}
    Let $\bgamma\in \bbR\times \bbR_{>0}$ then the evolution operator $\matT_\bgamma$ has a discrete spectrum formed by eigenvalues, with no accumulation point. The spectrum has the following characterization:
    \begin{enumerate}
        \item The eigenvalues are the solutions of the characteristic equation:
        \begin{align} \label{eq:char-eq}
            \sigma(\bgamma) &= \set{\lambda\in\bbC|\Delta_\bgamma(\lambda) = 0} \\
            \Delta_\gamma(\lambda) &:= \cf{\dfrac{\det\lp M_\bgamma(\lambda)\rp}{w_\bgamma(H)w_\bgamma(\theta)}=  \frac{f_1(\theta, \lambda, \bgamma)}{w_\bgamma(\theta)} - \frac{f_1(H, \lambda, \bgamma)}{w_\bgamma(H)}\exp(-\tau_0\lambda)}
        \end{align}
        % \mmnote{Invertire la posizione dei parametri di $f_1$: $f_1(\cdot, \bgamma, \lambda) \to f_1(\cdot, \lambda, \bgamma)$.} 
        where $f_1$ is the fundamental solution defined in Equation \eqref{eq:fund-sol}, such that $\matL_\bgamma f_1 = \lambda \cf{f_1}$, and $\matS_\bgamma f_1(\alpha) = 0$.
        % Here $w_\bgamma$ denotes the (rescaled) Wronskian:
        \item 0 is always an eigenvalue. 

        \item All eigenvalues have geometric multiplicity $=1$. %\mmnote{E i punti di transizione complesso-reale? Quella è la molteplicità algebrica.}.
    %     Let $\lambda_n =\lambda_n(\bgamma)\in \sigma(\bgamma)$ be an eigenvalue.  
    % Then, there exist holomorphic functions:
    % $s \mapsto \bphi(s, \bgamma) \in \bbD_\bgamma$ and $s \mapsto \bpsi(s, \bgamma) \in \bbD^*$
    % such that
    % \begin{align}
    %     s \mapsto \lp\matT_\bgamma -\lambda \rp^{-1} - \frac{\bphi(s, \bgamma)\otimes \bpsi(s, \bgamma)}{(s-\lambda_n)^{m_n}}   
    % \end{align}
    % is holomorphic in a neighborhood of $\lambda_n$. Here, $m_n$ denotes the algebraic multiplicity of the eigenvalue $\lambda_n$.
    \end{enumerate}
    We will denote the eigenvalues as: $ \sigma(\bgamma) = \set{\lambda_n(\bgamma)| n\in \bbN}\subset \bbC$  ordered in decreasing lexicographic order, and repeated with the respective algebraic multiplicity $m_n$. 
    \begin{proof}
    \cf{By Lemma~\ref{thm:spectrum-bep} the operator function $\tmatT_\bgamma-\lambda$ is globally equivalent to the characteristic matrix $M_\bgamma(\lambda)$; hence $\lambda\in\sigma(\bgamma)$ if and only if $\det M_\bgamma(\lambda)=0$, and evaluating the determinant with $\matS_\bgamma f_1(\alpha)=0$ yields the characteristic function \eqref{eq:char-eq}. 
    In particular $f_1(v, 0, \bgamma) = w_\bgamma(v)$ solves $\matL_\bgamma f=0$ and satisfies the boundary conditions, so $\Delta_\bgamma(0)=0$ and $0$ is an eigenvalue. 
    The entries of $M_\bgamma$ are entire in $\lambda$, so $\Delta_\bgamma:\bbC\to\bbC$ is entire; moreover $\Delta_\bgamma'(0)\neq0$ (see the proof of Proposition~\ref{prop:zero-simple}), whence $\Delta_\bgamma\not\equiv0$. As $\tmatT_\bgamma$ is a Fredholm operator function (Lemma~\ref{thm:spectrum-bep}), $\sigma(\bgamma)$ coincides with the zero set of the non-trivial entire function $\Delta_\bgamma$ and is therefore discrete with no accumulation point, each eigenvalue $\lambda_n$ having finite algebraic multiplicity $m_n$ equal to the order of the corresponding zero.

    Finally, if some $\lambda_n$ had geometric multiplicity %\mmtext{$m_n>1$}
    all minors of $M_\bgamma(\lambda_n)$ would vanish, forcing $f_1(\theta,\lambda_n,\bgamma)=f_1(H,\lambda_n,\bgamma)=0$; the boundary conditions then give $f_2(\theta,\lambda_n,\bgamma)=f_2(H,\lambda_n,\bgamma)=0$, whence $w_\bgamma(\theta)=0$, a contradiction. 
    Hence every eigenvalue has geometric multiplicity $1$.}

    \end{proof}
    \end{theorem}

Theorem \ref{thm:spectrum-prop} characterizes the spectrum of the evolution operator $\matT_\bgamma$. 
The eigenvalues are the solution of a characteristic equation, extending previous results. 
Indeed, at the limit $\tau_0 \to 0$, this equation reduces to the known expression \mmtext{\cite{Knight1996,Mattia2002},} obtained in the absence of an absolute refractory period. 
The characteristic equation can be rewritten using the eigenfunctions of the adjoint problem as: 
    \begin{align}
        \Delta_\bgamma(\lambda) = \cf{\psi_1(\theta, \lambda, \bgamma) - e^{-\tau_0\lambda}\psi_1(H, \lambda, \bgamma)}
    \end{align}
    Where $\psi_i = w_\bgamma^{-1} f_i$ are the fundamental solutions of the equation $\matL^*_\bgamma \psi_i =\lambda \psi_i$, $i\in\{1,2\}$, such that $\psi_i'(\alpha, \lambda, \bgamma) = -D^{-1}w_\bgamma^{-1} S_i(\alpha, \lambda, \bgamma) = \cf{-D^{-1}}\delta_{i2}$. 

\subsection{Root functions}

As the characteristic matrix completely determines the spectral properties of the boundary eigenvalue operator function, it can be used to construct the root functions associated with the eigenvalue problem.

\begin{lemma}\label{lem:root-functions}
There exist holomorphic functions
\[
s \mapsto \bphi(s,\bgamma) \in \bbD_\bgamma,
\qquad
s \mapsto \bpsi(s,\bgamma) \in \bbD^*,
\]
given by
\begin{align}
    \bpsi(s,\bgamma)
    &=
    \begin{bmatrix}
        \psi_1(\,\cdot\,,s,\bgamma)\\
        \psi_1(\theta,s,\bgamma)\exp(\,\cdot\, s)
    \end{bmatrix},
    \\
    \bphi(s,\bgamma) \label{eq:root-phi}
    &=
    \begin{bmatrix}
        f_1^-(\,\cdot\,,s,\bgamma)\Big(e^{-s\tau_0}\psi_2(H,s,\bgamma)-\psi_2(\theta,s,\bgamma)\Big)
        \cf{+f_2^-(\,\cdot\,,s,\bgamma)\Delta_\bgamma(s)}\\
        f_2^+(\,\cdot\,,s,\bgamma)\psi_1(\theta,s,\bgamma)
        -f_1^+(\,\cdot\,,s,\bgamma)\psi_2(\theta,s,\bgamma)\\
        \exp(-\,\cdot\, s)
    \end{bmatrix},
\end{align}
% where
% \begin{align}
%     \Delta_\bgamma(s)
%     :=
%     e^{-s\tau_0}\psi_1(H,s,\bgamma)-\psi_1(\theta,s,\bgamma).
% \end{align}
For every eigenvalue $\lambda_n \in \sigma(\bgamma)$, with algebraic multiplicity $m_n$, the map
\begin{align}\label{eq:principal-part-resolvent}
    s \mapsto
    \left(s-\matT_\bgamma\right)^{-1}
    -
    \frac{\bphi(s,\bgamma)\otimes\bpsi(s,\bgamma)}{\Delta_\bgamma(s)}
\end{align}
is holomorphic in a neighborhood of $\lambda_n$.
\end{lemma}

The functions $\bphi(s,\bgamma)$ and $\bpsi(s,\bgamma)$, properly normalized at each eigenvalue $\lambda_n$, form a biorthogonal canonical system of root functions (CSRF) for $\tmatT_\bgamma$ and its adjoint. In particular, the eigenfunctions of $\matT_\bgamma$ and $\matT_\bgamma^*$ corresponding to the eigenvalue $\lambda_n$ are given by $\bphi_n(\bgamma)=\bphi(\lambda_n,\bgamma)$ and $\bpsi_n(\bgamma)=Z_n^{-1}\bpsi(\lambda_n,\bgamma)$\footnote{For $n=0$, following the standard convention in the population-density approach, we define $\bphi_0(\bgamma)=Z_0^{-1}\bphi(0,\bgamma)$ so that $\bphi_0(\bgamma)$ coincides with the stationary probability distribution.}, while the associated vectors are obtained from the derivatives of these holomorphic root functions at $s=\lambda_n$.

\begin{remark}
\label{rem:norm-constant}
The normalization of $\bphi$, $\bpsi$, and of the characteristic function $\Delta_\bgamma$ is not unique, since these quantities are defined only up to nonzero scalar factors. 
In this work we choose the above normalization so that the firing rate of the primal root function satisfies
\begin{displaymath}
   N\bphi(s,\bgamma)=1.
\end{displaymath}
With this convention, for each eigenvalue $\lambda_n\in\sigma(\bgamma)$ of algebraic multiplicity $m_n$, the quantity
\begin{align}
    Z_n=
    \lim_{s\to\lambda_n}
    \frac{\Delta_\bgamma(s)}{(s-\lambda_n)^{m_n}}
\end{align}
is the leading coefficient in the local expansion of $\Delta_\bgamma$ at $\lambda_n$, so that
\begin{displaymath}
\Delta_\bgamma(s)=Z_n(s-\lambda_n)^{m_n}+o\big((s-\lambda_n)^{m_n}\big)
\qquad \text{as } s\to\lambda_n.
\end{displaymath}
In particular, if $\lambda_n$ is simple, then $Z_n=\Delta_\bgamma'(\lambda_n)$. 
This formula generalizes the expression obtained in \cite[Eq.~(28)]{vinci2024rosetta} by allowing for the presence of an absolute refractory period and of non-simple eigenvalues. 
As detailed in \cf{Appendix~\ref{sec:supp-adjoint}}, it is convenient to derive these root functions by considering the boundary eigenvalue problem associated with the adjoint, which leads to a reduced $3\times 3$ characteristic matrix and a simpler explicit construction.
\end{remark}

We refer to \cf{Appendix~\ref{sec:supp-adjoint}} for further details and derivations.

\subsection{Spectrum for LIF neurons}

\mmtext{To gain deeper insights into the spectrum of the evolution operator, we now numerically solve the characteristic equation for the standard ``leaky integrate-and-fire'' (LIF) neuron model, $ A(v) = -v $ \cite{tuckwell1988introduction, Gerstner2014}.}
% which is widely regarded as the 'gold standard' for spiking neural network models. \luca{insert citations on the importance of LIF neuron}. 
In this particular case, the eigenfunctions of the FP operator can be analytically expressed using confluent hypergeometric functions \cite{Ricciardi1988, brunel1999fast}.
To better understand the effect of the refractory period on the point spectrum, we start by analysing the structure of the eigenvalues \mmtext{for $\tau_0 = 0$.}
%relative to a neural population without refractory period ($\tau_0 = 0$). 

\begin{figure}
   \centering
   \includegraphics[width=0.99\linewidth]{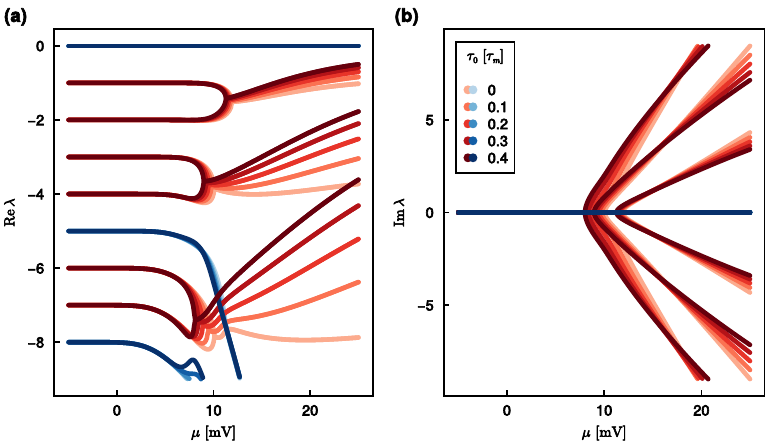}
   % ---- superseded caption, kept for reference (figures regenerated 2026-07-20) ----
% \caption{Eigenvalues $\lambda_n$ of the evolution operator for different values of refractory period $\tau_0$.
%    \mmtext{Current mean $\mu$ is varied at $0.05\,\text{mV}$ steps keeping $\DD = 8 \text{mV}^2$.
%    Red and blue branches are mixed and real $\lambda_n$, respectively.
%    Different shadings represent $\tau_0$ from $0 \, \tau_m$ (light) to $0.4 \, \tau_m$ (dark).
%    Parameters: $\tau_m = 1$,} \mmnote{values of $\alpha$, $H$ and $\theta$.}}
\caption{Eigenvalues $\lambda_n$ of the evolution operator for different values of refractory period $\tau_0$.
   \mmtext{Current mean $\mu$ is varied at \cf{$0.1\,\text{mV}$} steps keeping $\DD = 8 \text{mV}^2$.
   \cf{(a) Real part of $\lambda_n$; (b) imaginary part.}
   Red and blue branches are mixed and real $\lambda_n$, respectively.
   Different shadings represent $\tau_0$ from $0 \, \tau_m$ (light) to $0.4 \, \tau_m$ (dark).
   Parameters: $\tau_m = 1$\cf{, $H = 0$ mV, $\theta = 20$ mV, $\alpha = -5\theta = -100$ mV},}}
   \label{fig:spectrum-range}
\end{figure}

To achieve this, we numerically find the solutions of the characteristic equation \eqref{eq:char-eq} in the domain $[-10, 10]\times [-20i, 20i] \in \bbC$. 
\mmtext{We then keep fixed the variance $2 \DD$ of the input current and change its mean $\mu$ in the range $[-5,25] \, \text{mV}$ (Figure~\ref{fig:spectrum-range}). 
As we expected, $0$ is always an eigenvalue. 
The non-leading eigenvalues can be grouped into two classes: real (blue) and mixed (red).}
Real eigenvalues correspond to purely relaxational modes, while complex ones give rise to damped oscillations, with the imaginary part setting the oscillation frequency \mmtext{(Figure~\ref{fig:spectrum-range}-right). 
In real branches $\lambda_n$ remain real for all values of $\mu$ and tend to $-\infty$ as \mmtext{$\mu \to \theta/\tau_m$,} with the rate of decay depending on the noise level. 
According to \cite{Augustin2017, Mattia2021, vinci2024rosetta}), this set of eigenvalues coexist with the standard branches observed for $\alpha = H$ \cite{Mattia2002} where $\lambda_n$ undergo a real-to-complex transition. 
For low $\mu$, the eigenvalues are real, at a critical value, they meet and then split into complex conjugate pairs. 
Numerically, for $\tau_0=0$ this occurs for values of $\mu$ greater than $(H + \theta)/2$ \cite{Mattia2021}.}
As $\mu$ increases further, the real part stabilises while the imaginary part grows approximately linearly.
As $\tau_0$ increases, the onset of the real-to-complex transition in the mixed branches occurs at lower values of $\mu$. More notably, beyond this point, the real part of the eigenvalues increases with $\mu$, leading to more pronounced oscillations at higher $\mu$ -- a phenomenon absent when $\tau_0 = 0$. 
This suggests that the presence of a refractory period enhances the propensity for oscillatory dynamics.
Regarding the imaginary part, increasing $\tau_0$ reduces the slope of its growth with $\mu$, resulting in lower oscillation frequencies: \mmtext{a phenomenon related to the redunction of the firing rate $\nu$ as a larger $\tau_0$ leads to longer inter-spike intervals.}

\subsection{Defective eigenvalues}\label{sec:defective_eigenvalues}

While each eigenvalue admits a unique eigenfunction, not all eigenvalues are necessarily simple. 
\mmtext{This is the case of the real-to-complex transition for the mixed $\lambda_n$, where defective eigenvalues--those with algebraic multiplicity $m_n>1$--appear as} the characteristic equation admits a zero of order higher than one: 
% \begin{align}   
%  0 =& 
% \partial_\lambda\Delta_\bgamma(\lambda)|_{\lambda=\lambda_n} = \\ \notag
%  =& \partial_\lambda \psi_1(\theta, \lambda, \bgamma)|_{\lambda=\lambda_n} + \exp(-\lambda_n\tau_0)\left[\tau_0 \psi_1(H, \lambda_n, \bgamma) - \partial_\lambda \psi_1(H, \lambda, \bgamma)|_{\lambda=\lambda_n}\right]
% \end{align}
\begin{align}   
% 0 = &\Delta_\bgamma(\lambda) = \psi_1(\theta, \lambda, \bgamma) - e^{-\tau_0\lambda}\psi_1(H, \lambda, \bgamma) \notag\\
%  0 =& 
0 = \partial_\lambda\Delta_\bgamma(\lambda)   
 = \cf{\partial_\lambda \psi_1(\theta, \lambda, \bgamma) - e^{-\lambda\tau_0}\left[\partial_\lambda \psi_1(H, \lambda, \bgamma) - \tau_0 \psi_1(H, \lambda, \bgamma)\right]} \label{eq:def-eig-diff}
\end{align}
% These points coincide with the transitions in the mixed branches, where two real eigenvalues coalesce and subsequently give rise to a complex conjugate pair.

\begin{figure}
   \centering
   \includegraphics[width=0.99\linewidth]{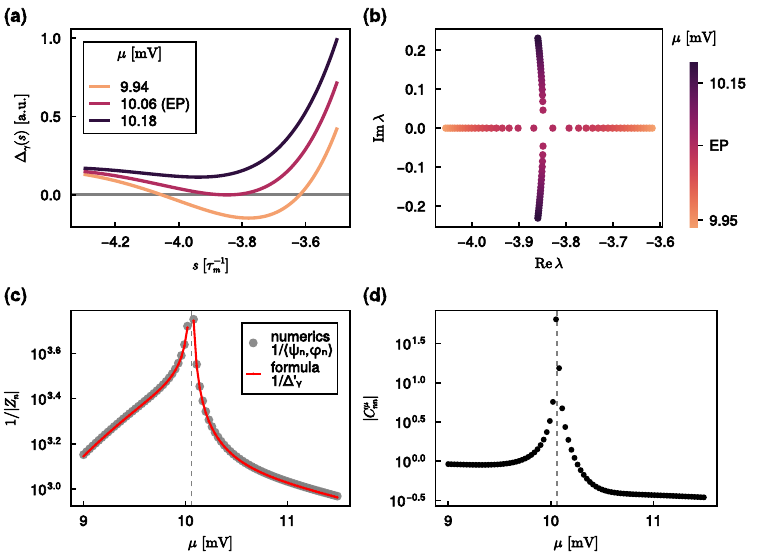}
   % ---- superseded caption, kept for reference (figures regenerated 2026-07-20) ----
% \caption{\mmtext{Defective eigenvalues at the real-to-complex transition in a mixed branch.
%    A, Example real-to-complex transition at varying $\mu$ (second mixed branch in Figure~\ref{fig:spectrum-range}) in the complex plane of the eigenvalues.
%    Light to dark points, increasing mean current $\mu$.
%    B, Real part of $\Delta_\bgamma(s)$ for the three $\mu$ around the transition.
%    Intersections with the x-axis are real $\lambda_n$ in (same as B).}
%    \mmnote{label degli assi troppo piccole.}}
\caption{\mmtext{Defective eigenvalues at the real-to-complex transition in a mixed branch.}
   \cf{(a) Real part of $\Delta_\bgamma(s)$ for three mean currents around the transition, $\mu = 9.94$, $10.06$ (exceptional point) and $10.18$ mV; intersections with the horizontal axis are real $\lambda_n$.
   (b) The same transition (second mixed branch in Figure~\ref{fig:spectrum-range}) in the complex plane of the eigenvalues; light to dark points, increasing $\mu$ from $9.94$ to $10.18$ mV.
   (c) Inverse leading coefficient $1/|Z_n|$ across the transition, evaluated from the biorthogonal scalar product $1/\langle\psi_n,\phi_n\rangle$ (grey) and from the closed form $1/\Delta_\bgamma'(\lambda_n)$ (red); the two agree away from the transition and both diverge at the exceptional point (dashed line), where $\Delta_\bgamma$ acquires a second-order zero.
   (d) Modulus of the self-coupling coefficient $C^\mu_{nn} = \langle\partial_\mu\psi_n,\phi_n\rangle$, which diverges at the same point: the modal decomposition ceases to be well defined at a defective eigenvalue.
   Parameters: $\DD = 8\,\text{mV}^2$, $H = 0$ mV, $\theta = 20$ mV, $\tau_0 = 0$, $\tau_m = 1$.}
   \mmnote{labels of $x$ and $y$ axes too small: increase font size, not lower than 6 pt.}}
   \label{fig:deg_eig}
\end{figure}

Figure~\ref{fig:deg_eig}\cf{b} illustrates such real-to-complex transition in the complex plane. 
\mmtext{By increasing the mean current $\mu$, the characteristic equation exhibits a second-order zero at the exceptional point, corresponding to a defective eigenvalue (Figure~\ref{fig:deg_eig}\cf{a}).}
The condition given by equation \eqref{eq:def-eig-diff} involves the derivative of the homogeneous solution $\psi_1$ with respect to the eigenvalue parameter. 
By differentiating the expression $\lp\matL^*_\bgamma - s\rp \psi_1(\cdot, s, \bgamma)$, and using the separation of variables formula, we arrive at the integral expression: 
 \begin{align}
     \partial_s\psi_1(v, s, \bgamma) = \cf{\psi_1(v)\int_\alpha^{v} f_2(u)\psi_1(u) du - \psi_2(v)\int_\alpha^{v} f_1(u)\psi_1(u) du} \, .
 \end{align}
An equivalent characterization for defective eigenvalues can be found by observing that: 
\begin{align}
    \langle \bpsi(s, \bgamma), \bphi(\lambda_n, \bgamma)\rangle = \cf{\frac{e^{(s-\lambda_n)\tau_0}\Delta_\bgamma(s)}{ s - \lambda_n}} \quad \lambda_n \in \sigma(\bgamma)
\end{align}
Therefore an eigenvalue $\lambda_n \in \sigma(\bgamma)$ is defective if and only if $\langle \bpsi(\lambda_n, \bgamma), \bphi(\lambda_n, \bgamma) \rangle = 0 $. 
This condition can be used to prove that $0$ is always a simple eigenvalue. 

\begin{proposition}\label{prop:zero-simple}
   $0\in \sigma(\bgamma)$ is a simple eigenvalue with eigenfunctions:
   \begin{align}
       \bphi_0(\bgamma) = (\phi_0(\cdot, \bgamma), \nu_0(\bgamma)) \quad  \bpsi_0(\bgamma) = (1, 1)
   \end{align}
Where 
\begin{align}
    \phi_0(v, \bgamma) = Z_0^{-1} \cb{\DD^{-1}} w_{\bgamma}(v) \int_{H\lor v }^\theta \!\!\! w_{\bgamma}^{-1}(u)du
\end{align}
\begin{proof}
    The eigenfunction expression follows by evaluating equation \eqref{eq:root-phi} at 0. To see that $0\in\sigma(\bgamma)$ is simple it is sufficient to observe that:
    $$\langle \bpsi(0, \bgamma), \bphi(0, \bgamma) \rangle = \cb{\DD^{-1}}\int_\alpha^\theta w_{\bgamma}(v) \int_{H\lor v }^\theta \!\!\! w_{\bgamma}^{-1}(u)du dv + \tau_0 > 0 $$
    \cb{This quantity is exactly $Z_0$, so the same computation gives $\|\bphi_0\|_{\bbL^1}=1$: the
    normalization of the stationary density follows from the definition of $Z_n$ rather than being
    imposed separately.}
\end{proof}
\end{proposition}
Stationary eigenfunctions are particularly important, as any stationary solution of the evolution equation \eqref{eq:abstract-ode} takes the form $\bphi_0(\bgamma) = (\phi_0(\cdot, \bgamma), \nu_0(\bgamma))$ with $\nu_0(\bgamma)$ giving the firing rate of the population.

Previous studies on the spectral decomposition of neuronal population dynamics \cite{Knight1996, Mattia2002, deniz2017solving, pietras2020low} have implicitly assumed that all eigenvalues are simple. 
This assumption is mathematically convenient, since simple eigenvalues depend smoothly on the system parameters $\bgamma$ and therefore give rise to smoothly varying spectral projectors and eigenfunctions.

\begin{remark}\label{rem:simple-eig-smooth}
Let $\bgamma_\star \in \bbR \times \bbR_{>0}$ and let $\lambda^{(0)} \in \sigma(\bgamma_\star)$ be a simple eigenvalue. Then, by the Implicit Function Theorem applied to the characteristic equation
\[
\Delta_\bgamma(\lambda) = 0,
\]
there exist neighborhoods $\bgamma_\star \in U \subset \bbR \times \bbR_{>0}$ and $\lambda^{(0)} \in V \subset \bbC$, together with a continuously differentiable function
\[
\Uplambda:U\to V,
\]
such that
\[
\{\Uplambda(\bgamma)\}=\sigma(\matT_\bgamma)\cap V
\qquad \forall \bgamma \in U.
\]
Moreover, differentiating the identity $\Delta_\bgamma(\Uplambda(\bgamma))=0$ yields
\begin{equation}\label{eq:lambda_sensitivity}
    \partial_\bgamma \Uplambda
    =
    -
    \frac{\partial_\bgamma \Delta_\bgamma(\Uplambda(\bgamma))}
    {\partial_\lambda \Delta_\bgamma(\Uplambda(\bgamma))}.
\end{equation}
If $\lambda \in \sigma(\bgamma)$ is defective, i.e. has algebraic multiplicity $m>1$, then it is a multiple root of the characteristic equation and therefore 
% \mmnote{Is this true only for $m=2$?}
% \lfnote{No, it is true for any $m>1$.}
\[
\partial_\lambda \Delta_\bgamma(\lambda)=0.
\]
Hence, at such points the representation \eqref{eq:lambda_sensitivity} breaks down, and the smooth parameterization of the eigenvalue branch generally fails.
\end{remark}

\subsection{Semigroup}\label{sec:semigroup_dissipativity}

The formal study of solutions to the evolution equation \eqref{eq:abstract-ode} can be approached via the theory of Markov semigroups associated with the stochastic process governing the membrane potential. 
In \cite{liu2022rigorous}, the authors established a formal link---valid for stationary currents ($\dot\bgamma = 0$) and in the absence of an absolute refractory period---between the integrate-and-fire process on the half-line and its corresponding Fokker--Planck PDE.

Related results were obtained in \cite{Grigorescu2002}, where Brownian motion on an interval with resetting at the origin was analyzed. 
The Markov semigroup for the process is defined on continuous functions over a ``figure-eight'' topology, where the endpoints are identified with the origin. 
These results were later extended to compact domains in $\mathbb{R}^n$ \cite{ben2009ergodic}, where the ergodicity of the semigroup and its convergence to equilibrium were studied through the spectral properties of its generator. 
Further extensions include processes with ``holding and jumping boundaries'' \cite{Peng2012DiffusionsWH}, where trajectories are held at the boundary for an exponentially distributed time before being reset. 
The spectral properties of such processes, in the driftless Brownian case, were analyzed in \cite{Leung2022}, where analytical bounds on the spectral gap were established.

While our setting shares similarities with these models, it differs in two biologically motivated aspects: the presence of a reflecting barrier at one endpoint and a deterministic holding time $\tau_0$. 
The techniques developed in previous works could, in principle, be adapted to account for the reflecting barrier and analyze the evolution equation \eqref{eq:abstract-ode} for stationary external currents. 
In this setting, the semigroup generated by the evolution operator uniquely determines the solution to the abstract Cauchy problem associated with the equation \eqref{eq:abstract-ode} in the case of an uncoupled population with a stationary external current. 
Here, we will briefly establish a generation theorem via dissipative operator theory, we refer to Appendix~\ref{app-dissipative} for the complete proofs of the statements.  
% \luca{rephrase: The techniques in previous works could be adapted taking care of the presence of reflecting barrier to study the evolution equation for stationary external currents. Here, we limit show a generation theorem using different techniques that will allow us to prove further properties of the spectrum of the generator. We will then conclude considering coupled cases through a perturbative approach} Although the techniques used in previous works could, in principle, be extended to our case by incorporating the reflecting boundary, our focus here is instead on directly characterizing the spectrum of the semigroup generator. 
% In this section, we outline the fundamental properties of the semigroups generated by the evolution operators, while the remainder of the paper is dedicated to the spectral analysis of the evolution operator $\matT_\bgamma$ and its adjoint.

We begin by showing that $\matT_\bgamma$ is a dissipative operator in $\bbL^1$. 

\begin{lemma} \label{lemma-dissipative}
The operator $\matT_\bgamma: \prescript{1}{}{\bbD_\bgamma} \subset \bbL^1 \to \bbL^1$ is dissipative and the following inequality holds: 
    \begin{align}\label{eq:diss-ineq}
    \quad \forall \mathbf{p}\in \prescript{1}{}{\bbD_\bgamma} \\ \notag
        Re\ \langle \mathrm{sign}(\mathbf{p}), \matT_\bgamma \mathbf{p}\rangle = Re \lp\int_\alpha^\theta \sign \lp p(v) \rp \matL_\bgamma p(v) dv \cf{-} \int_0^{\tau_0} \sign \lp p^r(\tau)\rp \partial_\tau p^r(\tau) d\tau \rp \le 0
    \end{align}
    % The equal sign holds iff:
    % \begin{align}
    %     \mathrm{arg }\ p(v) = \text{constant}
    %     % \mathrm{arg }\ p^r(0) = \mathrm{arg }\ p^r(\tau_0)\quad 
    %     \forall v \in (\alpha, \theta), h(v) > 0
    % \end{align}
\end{lemma}
We can then prove the following: 
\begin{proposition} \label{prop:eig-inequality}
    $\forall \pzp \in [1, \infty)$ The spectrum of the operator $\matT_\bgamma$ lies closed left-half plane: 
    \begin{align}\label{eq:sp-lefthalf}
        \sigma(\bgamma) \subset \set{s \in \bbC | Re(s) \le 0}
    \end{align}
    and $0=\lambda_0\in \sigma(\bgamma)$ is the only eigenvalue with non-negative real part. 
\end{proposition}
By Lumer-Philips Theorem \cite[Therem 3.15]{Engel1999OneparameterSF} we can immediately deduce 
\begin{theorem}
The evolution operator $\matT_\bgamma: \prescript{1}{}{\bbD_\bgamma}\subset \bbL^1 \to \bbL^1$ is the generator of a positive strongly continuous contraction semigroup $\lp\matF^1_\bgamma(t)\rp_{t\ge0}$ on $\bbL^1$.%\begin{proof}
\end{theorem}

Given a real positive initial datum $0<\bp_0\in \bbL^1$ the probability density $0<\bp_t = \matF_1(t) \bp_0$ is the unique positive solution of the abstract differential equation \eqref{eq:abstract-ode} for stationary current $\bgamma(t) \equiv \bgamma$.
By duality, the adjoint $\matT_\bgamma^*: \prescript{\infty}{}\bbD \subset \bbL^\infty \to \bbL^\infty$ generates a strongly continuous Markov semigroup when restricted on the closure of $\prescript{\infty}{}{\bbD_\bgamma}\subset \bbL^\infty$ \cite[p. 62-3]{Engel1999OneparameterSF} (see also \cite{Grigorescu2002, Peng2012DiffusionsWH}). 
In our case $\overline{\prescript{\infty}{}{\bbD_\bgamma}} = \mathrm{C}(P)$ is given by the continuous functions on a ``P'' shaped domain, union of the two distinct intervals $[\alpha, \theta]$, $[0, \tau_0]$ where we identify $H$ (resp $\theta$) in the first domain with $\tau_0$ (resp. $0$) in the second. 

The adjoint semigroup characterizes the relaxation of observables in uncoupled networks subject to a stationary external current. 
\mmtext{Extending this framework to coupled systems with time-dependent input is naturally done within a spectral setting, which underscores the need for a rigorous characterization of the spectrum and eigenfunctions. 
In the next section, we show how the spectral theory developed here can be used, through a perturbative approach, to derive the linear response and analyze the local stability of fixed points in coupled populations with $\tau_0 \geq 0$.}

% \begin{table}[htbp]
% \footnotesize
% \caption{Example table.}\label{tab:foo}
% \begin{center}
%   \begin{tabular}{|c|c|c|} \hline
%    Species & \bf Mean & \bf Std.~Dev. \\ \hline
%     1 & 3.4 & 1.2 \\
%     2 & 5.4 & 0.6 \\ 
%     3 & 7.4 & 2.4 \\ 
%     4 & 9.4 & 1.8 \\ \hline
%   \end{tabular}
% \end{center}
% \end{table}

\section{\mmtext{First-order perturbation theory}}
\label{sec:tf}

\mmtext{To extend our analysis beyond the stationary case, where $\bgamma$ is constant, here we study} the linearized dynamics around the stationary solution 
% $\bphi_0(\bgamma) = (\phi_0(\cdot, \bgamma), \nu_0)$ 
using first-order perturbation theory. 
% we recall that $\nu_0 = N(\bphi_0(\bgamma))$ is the stationary firing rate.
We start by considering an uncoupled population in which the moments of the external current are modulated by a small time-dependent factor, given by $\bgamma(t) = \bgamma_0 + \bgamma_1 \varepsilon(t)$, with $\bgamma_1 = (\mu_1, \DD_1)$. We can then write 
\begin{equation}
    \matT_{\bgamma(t)} = \matT_{0} + \matT_{1}\varepsilon(t) \quad
    \matB_{\bgamma(t)} = \matB_{0} + \matB_{1}\varepsilon(t) \quad
\end{equation}
Where we defined $\matT_{0} := \matT_{\bgamma_0}$, $\matB_{0} := \matB_{\bgamma_0}$ and $\matT_1:= \nabla_\mathbf{\bgamma} \matT_\mathbf{\bgamma}|_{\bgamma_0} \cdot \bgamma_1 =: ( \matL_1, 0)$ and $ \matB_1:= \nabla_\mathbf{\bgamma} \matB_{\bgamma}|_{\bgamma_0}\cdot\bgamma_1$, such that:
\begin{align}\notag
    \matT_1\bphi_0 &= \lp\matL_1 \phi_0, 0 \rp = \lp -\mu_1\phi'_0(v) + \DD_1\phi_0''(v) , 0 \rp \\
    \matB_1\bphi_0 &= 
    \begin{bmatrix}
    \mu_1\phi_0(\alpha) -  \DD_1\phi_0'(\alpha)\\
    0\\
    0 \\
    -\DD_1\lp \phi_0'(H^+) - \phi_0'(H^-)\rp\\
    -\DD_1\phi_0'(\theta, \bgamma)   \\
\end{bmatrix}
=  \begin{bmatrix}
    \mu_1\phi_0(\alpha) -\DD_1 \phi_0'(\alpha)\\
    0\\
    0 \\
    \nu_0 \DD_1/\DD_0\\
    \nu_0 \DD_1/\DD_0\\
\end{bmatrix}
\end{align}

To avoid unnecessary clutter, in the rest of the section, we remove the dependence on $\bgamma$ on the symbols, which are evaluated as a fixed value $\bgamma=\bgamma_0$. Expanding the solution as $\mathbf{p}_t = \bphi_0 + \mathbf{p}^{1}_t + o(\varepsilon(t))$ and neglecting higher-order terms, the first-order perturbation $\mathbf{p}^{1}_t$ satisfies the following evolution equation:

\begin{align}
    \begin{cases}\label{eq:t-pert}
        \partial_t \mathbf{p}^{1}_t - \matT_{\bgamma}\mathbf{p}^{1}_t&=  \matT_1\bphi_0 \varepsilon(t) \\
        -\matB_0\mathbf{p}^{1}_t &=  \matB_1\bphi_0\varepsilon(t)
    \end{cases}
\end{align}
 
Equation \eqref{eq:t-pert} describes a linear system with an external forcing. As we are interested in the long time behaviour, for simplicity we will always assume the initial condition $\mathbf{p}_0^{1} = \mathbf{0}$. The main quantity of interest, in this context, is the frequency domain response of the population to the stimulus, given by the Laplace transform of \mmtext{the first-order perturbation of the firing rate} $\nu_1(t) = N(\mathbf{p}^{1}_t)$. 
\mmtext{This quantity is the so-called ``transfer function''}\footnote{We indicate with $\widehat{\cdot}$ the Laplace transform of a time-dependent function.}
\begin{align}
    \mathcal{H}(s):= \frac{\widehat{\nu}_1(s)}{\widehat{\varepsilon}(s)} = \frac{N\lp\widehat{\mathbf{p}}^{1}(s)\rp}{\widehat{\varepsilon}(s)}
\end{align}
The Laplace transform of the solution to the forced problem \eqref{eq:t-pert} is given by the resolvent of the boundary eigenvalue problem, as:
\begin{align}
    \begin{cases}\label{eq:s-pert}
         \lp \matT_0 - s\rp \hatp(s) &=  -\matT_1\bphi_0\hat{\varepsilon}(s) \\
        \matB_0\hatp(s) &=  -\matB_1\bphi_0\hat{\varepsilon}(s)
    \end{cases}
\end{align}
Such that:
\begin{equation}
    \hatp(s) = -\lp\matT_0 - s\rp^{-1} \matT_1\bphi_0\hat{\varepsilon}(s) - \matZ(s)M^{-1}(s)\matB_1\bphi_0\hat{\varepsilon}(s) \quad s\in\bbC\setminus\sigma(\bgamma)
\end{equation} 
\begin{lemma}
    \mmnote{In this lemma I replaced $\Delta \to \Delta_\bgamma$, check if this is correct}
    Let $\hatp(s)$ solution of the system \eqref{eq:s-pert} we have that the transfer function is given by:
    \begin{align} \label{eq:tf-wronskian}
        \mathcal{H}(s) =& \cf{\frac{1}{\Delta_\bgamma(s)}\lp \frac{\{f_1,Q\}_\bgamma(\theta, s)}{w_\bgamma(\theta)} - \frac{\{f_1,\jump{Q}\}_\bgamma(H, s)}{w_\bgamma(H)}\rp} + \\
        % - \mathbf{g}\cdot \frac{\psi_1(\theta, s, \bgamma)\matB^{(1)}_\bgamma \bphi_0(\bgamma)}{\Delta_\bgamma(s)} |
         &+\frac{\nu_0 \DD_1}{\DD_0}\frac{\psi_1(\theta, s) - \psi_1(H, s)}{\Delta_\bgamma(s)} \cf{- \frac{\mu_1\phi_0(\alpha) -\DD_1 \phi_0'(\alpha)}{\Delta_\bgamma(s)}}\notag \, ,
    \end{align}
where $Q(\cdot, s)$, is any particular solution to the problem $(s - \matL_0)Q(\cdot, s) = \matL_1 \phi_0$. $\cf{\{f,h\}_\bgamma} = f\matS_\bgamma h - h\matS_\bgamma f$, and $\jump{Q}(v, s) = \lim_{u\to 0} Q(v+|u|, s) - Q(v-|u|, s)$. 
\begin{proof}
    We begin by observing that $Q(\cdot, s)$, is any particular solution to the problem $(s - \matL_0)Q(\cdot, s) = \matL_1 \phi_0$ we have that $\widetilde{\mathbf{Q}}(s) = [Q(\cdot, s), 0] \in \bbW^\pzp$ is a particular solution to the equation  $\lp \matT_0 - s\rp \hatp(s) =  -\matT_1\bphi_0\hat{\varepsilon}(s)$. We then search for a solution of the form $\hatp(s) = \matZ(s)\bg \cf{+} \widetilde{\mathbf{Q}}(s)$ where $\bg \in \bbC^5$ is determined by the boundary conditions:
    \begin{align}
        B_0\lp \matZ(s)\bg \cf{+} \widetilde{\mathbf{Q}}(s)\rp = -\matB_1\bphi_0 \Rightarrow \\
        \bg =  \cf{-}M_0^{-1}B_0\widetilde{\mathbf{Q}}(s) - M_0^{-1}\matB_1\bphi_0 \label{eq:system-tf}
    \end{align}
     Considering that $N\hatp(s) = N \matZ(s) \bg \cf{+} N\widetilde{\mathbf{Q}}(s) = g_5$\cf{, the fifth entry of $\bg$, i.e.\ the coefficient of the refractory component of $\hatp$}.
Solving the linear system \eqref{eq:system-tf} after some computation, we arrive at equation \eqref{eq:tf-wronskian}. 
\end{proof}    
\end{lemma}
\begin{remark}
    If we use the particular solution given by the variation of parameters formula (see \mmtext{Remark~\ref{rm:var-par}}), we obtain an expression dependent on the eigenfunctions: 
    \begin{align}\label{eq:tf-old}
        \cf{\frac{1}{\Delta_\bgamma(s)}\lp\frac{\{f_1,Q\}_\bgamma(\theta, s)}{w_\bgamma(\theta)} - \frac{\{f_1,\jump{Q}\}_\bgamma(H, s)}{w_\bgamma(H)}\rp} = \frac{\int_\alpha^\theta \psi_1(v, s) \matL_1 \phi_0(v) dv}{\Delta_\bgamma(s)} = \frac{\langle \bpsi(s), \matT_1\bphi_0\rangle}{\Delta_\bgamma(s)}
    \end{align}
We refer to \cite[Section III.B]{vinci2024rosetta} for the details in the derivation. 
\end{remark}

This expression generalizes the known form of the transfer function for integrate-and-fire neurons without an absolute refractory period. 
The term in \eqref{eq:tf-old} corresponds to the standard contribution derived in \cite[Eq. 49]{vinci2024rosetta}, with the important distinction that the characteristic equation in the denominator, \cf{$\Delta_\bgamma(s)$ defined in \eqref{eq:char-eq}}, now depends explicitly on the absolute refractory period. 
In addition to this term, two further contributions appear due to the explicit dependence of the boundary conditions on $\bgamma$, which was not fully accounted for in previous literature \mmtext{\cite{brunel1999fast,Lindner2001,Richardson2007,Schuecker2015,vinci2024rosetta}}.
The first additional term arises from the modulation of the noise intensity. 
In the limit $\tau_0 \to 0$, it reduces to $\nu_0 \DD_1 / \DD_0$. 
We expect similar corrections to appear in more general integrate-and-fire models with state-dependent noise, such as conductance-based neurons. 
The second term reflects the presence of a reflecting barrier at $\alpha$. 
In models with a confining force field $A(v)$, this term vanishes as $\alpha \to -\infty$, but it could play a significant role in the \mmtext{so-called VLSI integrate-and-fire (VIF) models, where $A(v)$ is constant and $\alpha$ is finite \cite{Fusi1999, Mattia2002, Mattia2004}}. 
% The investigation of these effects will be the subject of future work. 
% Here, we will focus on analyzing the effect of the absolute refractory period on the transfer function for the LIF model.

\subsection{Transfer function of LIF neurons with refractory period}

In the case of \mmtext{LIF neurons,} the particular solution $Q(v, s)$ can be expressed in terms of derivatives of the stationary state $\phi_0(v)$ \cite{brunel1999fast}. 
This allows for a further simplification of the term \eqref{eq:tf-old}, leading to a non-integral expression for the transfer function \cite{vinci2024rosetta}. 
By isolating the contributions due to the modulation of drift $\mu$ and noise $D$, the transfer function can be written as:
$\mathcal{H} = \mu_1\mathcal{H}_\mu + \DD_1\mathcal{H}_{D}$.
This yields the following expressions for the frequency response of a population of LIF neurons with an absolute refractory period:
\begin{align}\notag
\mathcal{H}_{\mu}(s) &= \cb{\frac{\nu_0}{s+1}\,\frac{\partial_v\psi_1(\theta, s) - \partial_v\psi_1(H, s)}{\psi_1(\theta, s) - e^{-s\tau_0}\psi_1(H, s)} - \frac{s}{s+1}\,\frac{\phi_0(\alpha)}{\psi_1(\theta, s) - e^{-s\tau_0}\psi_1(H, s)}}\\
    \mathcal{H}_{D}(s) &= \cb{\frac{\nu_0\big(s[\psi_1(\theta, s) - \psi_1(H, s)] + (\theta\!-\!\mu_0)\partial_v\psi_1(\theta, s) - (H\!-\!\mu_0)\partial_v\psi_1(H, s)\big)}{\DD_0(s+2)\left(\psi_1(\theta, s) - e^{-s\tau_0}\psi_1(H, s)\right)}} \notag \\\notag
    &\quad \cb{+ \frac{s}{s+2}\,\frac{\phi_0'(\alpha)}{\psi_1(\theta, s) - e^{-s\tau_0}\psi_1(H, s)}}
\end{align}
We refer to \cite{vinci2024rosetta} for details in the derivation.

\begin{figure}
    \centering
    \includegraphics[width=0.99\linewidth]{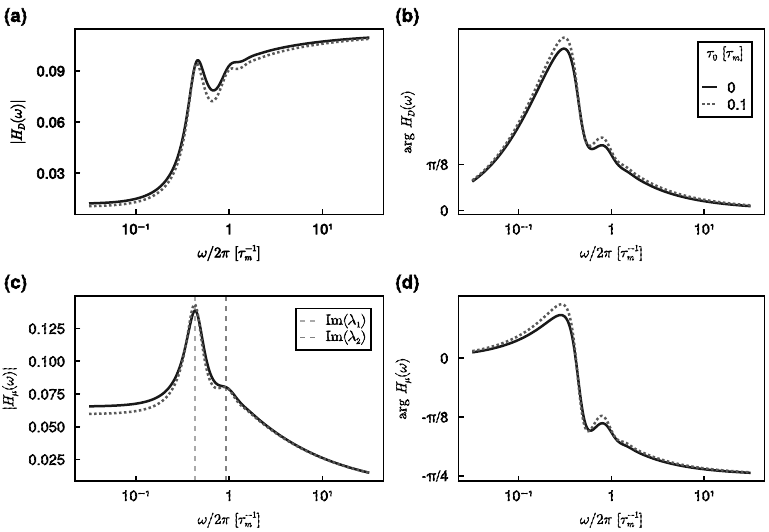}
    % ---- superseded caption, kept for reference (figures regenerated 2026-07-20) ----
% \caption{\mmtext{Module (left) and phase (right) of the transfer function in response to a drift (top) and noise (bottom) sinusoidal perturbation with (dashed line, $\tau_0 = 0.1 \tau_m$) or without (solid lines, $\tau_0 = 0$) absolute refractory period.
%    Vertical dashed lines, $\mathrm{Im} \lambda_n \simeq 2\pi n \nu_0$}. \mmnote{write network parameters and $\nu_0$.}}
\caption{\mmtext{Module (left) and phase (right) of the transfer function in response to a \cf{noise ($\mathcal{H}_D$, top) and drift ($\mathcal{H}_\mu$, bottom)} sinusoidal perturbation with (\cf{dotted line}, $\tau_0 = 0.1 \tau_m$) or without (solid lines, $\tau_0 = 0$) absolute refractory period.
   Vertical dashed lines, $\mathrm{Im} \lambda_n \simeq 2\pi n \nu_0$\cf{; the abscissa is the ordinary frequency $\omega/2\pi$ in units of $\tau_m^{-1}$}}. \cf{Parameters: $\mu = 21$ mV, $\DD \simeq 3.56\,\text{mV}^2$, $H = 0$ mV, $\theta = 20$ mV, $\alpha = -3\theta = -60$ mV, $\tau_m = 1$, $\nu_0 = 0.4\,\tau_m^{-1}$.}}
    \label{fig:TF}
\end{figure}

Figure~\ref{fig:TF} shows the effect of the absolute refractory period on the transfer function. 
As $\tau_0$ increases, the peaks in the drift-modulated response peaks become more pronounced at frequencies \mmtext{$\omega/2\pi = \text{Im}(\lambda_n)/2\pi \simeq n \nu_0$, while the phase of the response is mildly delayed.
This follows from the fact that the inter-spike interval distribution becomes relatively less variable as its mean increases, leading to more regular firing  and a smaller coefficient of variation.}
As we will see, this amplification influences the emergence of limit cycles in coupled populations.

\subsection{Linear stability for \mmtext{populations of interacting neurons}}
%  firing rate\footnote{Here $\ast$ denotes the convolution, such that $(p_D\ast \nu)(t) = \int_0^\infty \nu(t-u)p_D(u) du$}:
% \begin{align}
%     \mu(t) &= K J (p_D\ast \nu)(t) + \mu^{ext} \\
%     \DD(t) &= KJ^2 (p_D\ast \nu)(t) + \DD^{ext}
% \end{align}
% where $J$ denotes the synaptic efficacy, $K$ the number of presynaptic contacts, and the $p_D$ probability distribution of axonal delays, with support in $\mathbb{R}_{\ge 0}$. 
% As a result,
% \luca{oppure}
We now consider a \mmtext{population of interacting neurons.} 
Under the `extended' mean-field approximation \mmtext{\cite{Amit1991, Amit1997, brunel1999fast, brunel2000dynamics, Mattia2002}}, the drift and the diffusion coefficients of the Fokker-Planck equation depend self-consistently on past values of the firing rate, making the PDE nonlinear: 
\begin{align}
    \bgamma(t) &= \bgamma\cdot\nu(t-\delta) + \bgamma_0 
\end{align} 
where $\delta \ge 0$ is a deterministic delay term and $\bgamma_0$ represents a stationary external current. 
Fixed points of the dynamics are given by solutions of the form $\mathbf{p}_t = \bphi_0(\bgamma)$, where the moments of the current satisfy the self-consistency equation $\bgamma = N(\bphi_0(\bgamma)) + \bgamma_0$. 
Notice that since now the evolution equation is nonlinear \mmtext{(i.e., $\bgamma$ depends on the firing rate, and hence on $\mathbf{p}_t$)} the system may display multistability. 
We can use the linearized first-order dynamics \mmtext{\eqref{eq:t-pert} to study the stability of these fixed points by considering $\varepsilon(t) = \nu_1(t-\delta) = N(\mathbf{p}^{1}_{t-\delta})$}. 
\cf{The delay $\delta$ turns the linearized dynamics into a delay differential equation, whose stability, in the framework of delay semigroups \cite{Batkai2005, Engel1999OneparameterSF}, can be determined by the eigenvalues of the} perturbed boundary eigenvalue problem $\tmatT_0 + \mathcal{K}_1\in \hol{\bbC}{\linear{\bbW^\pzp}{\bbL^p \times \bbC^5}}$, where $\tmatT_0 = \tmatT_{\bgamma_0}$ and  $\mathcal{K}_1 = \mathcal{K}_1 = \lp \mathcal{K}^{D}_1, \mathcal{K}^{R}_1\rp  $ is a rank-1 perturbation given by:
    \begin{align}
    \mathcal{K}^{\matT}_1(s) = \matT_1\bphi_0\cdot e^{-s\delta}N \quad
    \mathcal{K}^{\matB}_1(s) = \matB_1\bphi_0\cdot e^{-s\delta}N
\end{align}
\begin{theorem}
    $\tmatT_0 + \mathcal{K}_1$ is a Fredholm operator valued function such that:
    \begin{align}
       \tmatT_0 + \mathcal{K}_1 = \mathcal{C}(s)
        \begin{bmatrix}
            M(s) + M_1(s)& \mathbf{0} \\
            \mathcal{K}^{D}_1(s)\matZ(s)& \mathbf{Id}_{\mathbb{L}^\pzp}
        \end{bmatrix}\mathcal{D}(s)
    \end{align}
    Where $\mathcal{D}=\mathcal{D}_{\bgamma_0}$, $\mathcal{C}= \mathcal{C}_{\bgamma_0}$, $\matZ = \matZ_{\bgamma_0}$ are the same as Lemma \ref{thm:spectrum-bep} and $M_1(s)$ is given by: 
    \begin{align}
       M_1(s) = 
        \left[
\begin{array}{cccc|c}
0 & 0 & 0 & 0 &  \\
0 & 0 & 0 & 0 &  \\
0 & 0 & 0 & 0 &  \lp \matB_1\bphi_0 - \matB_0 Q(s) \rp\cdot e^{-\delta s}\\
0 & 0 & 0 & 0 &  \\
0 & 0 & 0 & 0 &  \\
\end{array}\right]
    \end{align}
where $Q(s) = \mathcal{Q}(s)\matT_1\bphi_0$ is the particular solution of the equation $\tmatT^D_0(s) Q = \matT_1\bphi_0$ given by separation of variables (see Remark \ref{rm:var-par}). Then the spectrum of 
 $\tmatT_0 + \mathcal{K}_1$ is given by the solution to the characteristic equation:
 \begin{align}\label{eq:stable-coupled}
     \sigma\lp \tmatT_0 + \mathcal{K}_1 \rp = \{s\in\bbC | \mathcal{H}(s)= \cf{e^{\delta s}}\}
 \end{align}
 \begin{proof}
     By Theorem 1.11 in \cite{mennicken2003non} we have that:
     \begin{align}
         \mathcal{C}(s)^{-1} = \begin{bmatrix}
                         - B_0\matQ(s) & \mathbf{Id}_{\bbC^5} \\
                          \mathbf{Id}_{\bbL^\pzp} & \mathbf{0}
                         \end{bmatrix} \quad \matD(s)^{-1} = \lp \matQ(s), \matZ(s) \rp
     \end{align}
     Then:
     \begin{align}
         \tmatT_0 + \mathcal{K}_1 = \mathcal{C}
        \begin{bmatrix}
            M  - B_0\matQ\mathcal{K}^D_1\matZ + \mathcal{K}^R_1\matZ &  - B_0\matQ(s)\mathcal{K}^D_1\matQ + \mathcal{K}^R_1\matQ\\
            \mathcal{K}^{D}_1\matZ& \mathcal{K}_1^{D}\matQ + \mathbf{Id}_{\mathbb{L}^\pzp}
        \end{bmatrix}\mathcal{D}
    \end{align}
    We observe that, as the firing rate of the particular solution is  \mmtext{$N\lp\matQ(s)\rp = 0$,} we have $\mathcal{K}^D_1\matQ = 0$ and $\mathcal{K}^R_1\matQ = 0$. 
    From the definition $B_0\matQ\mathcal{K}^D_1(s)\matZ(s) = B_0 Q(s)\cdot e^{-\delta s}N\lp\matZ(s)\rp$ and $\mathcal{K}^R_1\matZ = B_1 \bphi_0\cdot e^{-\delta s}N\lp\matZ(s)\rp$ such that $ -B_0\matQ\mathcal{K}^D_1\matZ + \mathcal{K}^R_1\matZ = M_1$. The spectrum of the perturbed boundary value problem $\tmatT_0 + \mathcal{K}_1$, is then given by the characteristic equation $0 = \det(M(s) + M_1(s))$, computing this matrix determinant we arrive at equation \eqref{eq:stable-coupled}, \mmtext{generalizing previous results for $\tau_0 = 0$ \cite{brunel1999fast, brunel2000dynamics, Mattia2002}.}
 \end{proof}
\end{theorem}

\subsection{Emergence of limit cycles in excitatory populations}

We apply this stability analysis to a \mmtext{population of recurrently coupled excitatory LIF neurons} with delay.
We consider a fixed point that is stable in the absence of refractoriness ($\tau_0=0$) and numerically compute the spectrum of the perturbed boundary eigenvalue problem as $\tau_0$ increases \mmtext{(Figure~\ref{fig:poles_exc}a).}

Following \cite{Mattia2002}, the spectrum typically contains two types of poles: \textit{diffusion poles}, which originate from the intrinsic Fokker-Planck spectrum, and \textit{transmission poles}, which arise from the delay-induced feedback.
As the refractory period $\tau_0$ increases, a pair of complex conjugate diffusion poles moves towards the right half-plane.
At a critical value of $\tau_0$, these poles cross the imaginary axis (Hopf bifurcation), destabilizing the fixed point and leading to the emergence of a stable limit cycle \mmtext{(Figure~\ref{fig:poles_exc}b).} \cf{We note that the spectral analysis proves only the loss of linear stability of the fixed point; the emergence of a stable limit cycle is verified numerically and not proven here, as this would require a nonlinear bifurcation analysis beyond the present scope.}

This highlights the possibility of inducing stable oscillations in the positive drift regime by destabilizing diffusion poles rather than transmission poles, as previously described in \cite{Treves1993, Mattia2002}. 
Whereas transmission delays act on the transmission poles, the refractory period introduces a delay in the ``mass'' \mmtext{reinjection} that specifically influences the diffusion poles associated with the single-neuron dynamics. 
This confirms that the refractory period acts as a singular perturbation that can fundamentally alter the macroscopic network dynamics from a stable focus to a limit cycle oscillator.

\begin{figure}[t]
    \centering
    \includegraphics[width=1.0\linewidth]{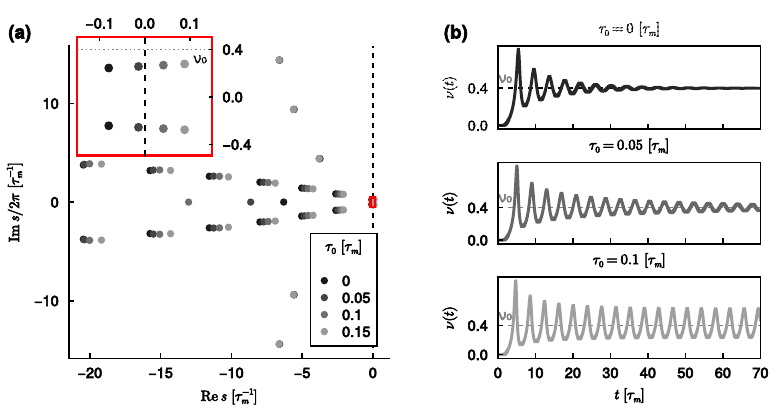}
    % ---- superseded caption, kept for reference (figures regenerated 2026-07-20) ----
% \caption{Locus of the leading eigenvalues (poles) of the coupled system as the refractory period $\tau_0$ is varied. The crossing of the imaginary axis by the diffusion poles indicates a Hopf bifurcation and the onset of global oscillations. Same parameters as Figure \ref{fig:TF}}
\caption{\cf{(a)} Locus of the leading eigenvalues (poles) of the coupled system as the refractory period $\tau_0$ is varied\cf{; the inset magnifies the pair closest to the imaginary axis}. The crossing of the imaginary axis by the diffusion poles indicates a Hopf bifurcation and the onset of global oscillations. \cf{(b) Population firing rate $\nu(t)$ from the corresponding Fokker--Planck simulations at $\tau_0 = 0$, $0.05$ and $0.1\,\tau_m$ (top to bottom): the damped relaxation towards $\nu_0$ (dashed) at $\tau_0 = 0$ gives way to a sustained limit cycle once the diffusion pair has crossed.} Same parameters as Figure \ref{fig:TF}\cf{, with $K = 1000$, $J = 0.012$ ($KJ = 12$) and transmission delay $\tau_D = 0.2\,\tau_m$}
    \mmnote{Complete the missing parameters of this caption. Label size must have font size not smaller than 6 pt. Increase the vertical size of plot in b. Plot diffusion poles and transmission poles with different symbols (circles and diamonds, for instance). In the inset of a you used circles while in rest of the panel diamonds.}}
    \label{fig:poles_exc}
\end{figure}

\section{Discussion}
% \luca{rivedi questa discussione anche in relazione ai nuovi risultati sulla analiticità}
\label{sec:discussion}

In this work, we have presented a rigorous mathematical framework for the spectral analysis of neuronal population dynamics with an absolute refractory period. To the best of our knowledge, this is the first work to systematically apply the theory of boundary eigenvalue problems \cite{mennicken2003non} to the Fokker-Planck operator governing population density dynamics, moving beyond the heuristic treatments often found in the literature.

Our central contribution is the analytical characterization of the spectrum, the resolvent, and the eigenfunctions of the evolution operator. 
By extending the state space to include the refractory history, we recovered a Markovian description amenable to rigorous operator theoretic analysis. 
We proved the dissipativity of the generator (Section\ref{sec:semigroup_dissipativity}), a fundamental property ensuring physical stability and the existence of a contraction semigroup. 
This formally demonstrates that a decoupled population always relaxes to \mmtext{a unique stationary state, as conjectured in \cite{Mattia2002}.} 
% Importantly, this dissipativity holds generally for the Fokker-Planck operator with reinjection boundary conditions, regardless of whether the refractory period $\tau_0$ is finite or zero, establishing a robust basis for the existence of solutions in both regimes.

A key result of our analysis is the identification and characterization of defective eigenvalues (Section \ref{sec:defective_eigenvalues}). 
These spectral singularities, \mmtext{observed numerically in different types of IF neurons and associated with potential instabilities in the numerical integration of spectral decompositions \cite{Mattia2002, Augustin2017},} are structural features of the operator for which the algebraic multiplicity exceeds the geometric one.
We have shown that they correspond to exceptional points in the parameter space where oscillatory modes emerge from the coalescence of relaxational modes. 
Properly accounting for the associated generalized eigenfunctions is essential for a complete description of the system's transient dynamics near these critical points.

\cf{For simple eigenvalues the coupling coefficients used in spectral \mmtext{decompositions} \cite{Knight1996, knight2000dynamics, Mattia2002}, $C_{nm}^\bgamma = \langle \partial_\bgamma \bpsi_n(\bgamma), \bphi_m(\bgamma)\rangle$, are well defined and the dynamics can be projected onto the time-dependent eigenbasis $\{\bphi_n(\bgamma(t))\}_n$. 
However, at defective eigenvalues this smooth dependence is lost, and a complete biorthogonal system must include the associated functions $\bphi_n^{(k)}$, $\bpsi_n^{(k)}$ generated by the Taylor expansion of the canonical system of root functions.}

Furthermore, our rigorous treatment of the boundary conditions has \mmtext{generalized} the linear response theory for neuronal populations. 
We derived an exact expression for the transfer function (Section \ref{sec:tf}) that explicitly accounts for boundary conditions modulated by external parameters.
This correction resolves ambiguities in previous derivations and reveals new contributions to the population response---specifically terms related to noise modulation at the threshold---that are likely to be significant in conductance-based or state-dependent noise models \cite{Richardson2004, Lindner2004}.

\mmtext{Our results provide a} foundation for a deeper mathematical exploration of \mmtext{the population density approach.} 
Two major open questions remain: the rigorous proof of the completeness of the eigenbasis and \mmtext{a detailed characterization of the} regularity properties of the generated semigroup (e.g., analyticity). 
\cf{For the class of problems considered here, completeness can be established through asymptotic estimates on the resolvent, in the sense of the Birkhoff and Stone regularity for the boundary value problem \cite{mennicken2003non}.}
\mmtext{A natural open question is whether the generated semigroup is analytic, which would imply stronger smoothing and stability properties for the population dynamics.}
\mmtext{Both of these developments, which would place} \cf{the spectral-decomposition methods widely used in computational neuroscience on a rigorous footing, will be the topic of future work.}

% In the next chapter, we will tackle this challenge directly by establishing the Birkhoff regularity of the boundary value problem and proving the sectoriality of the generator.

% \section{Discussion}
% \luca{Main ideas:
% \begin{itemize}
%     \item First work that studies the spectrum of the FP operator population density dynamics using rigorous spectral theory
%     \item Analytical characterisation of the spectrum and the resolvent, dissipativity of the generator. 
%     \item We show that how we can genralize previous results to describe neural population with arp
%     \item We can analytically characterise defective eigenvalues, and stress why these points are important
%     \item Our setting allows us to correctly describe the transfer function of the system, correctly taking into account BC that depend on external parameters 
%     \item Comment on the presence of finite $\alpha$, biological plausibility, and limit $\alpha\to -\infty$ for "confining" potential. 
%     \item Laid the foundations for Future work: completeness of the basis, study further regularity properties of the semigroup (eg. analiticy).  
% \end{itemize}}
% \label{sec:conclusions}

% Some conclusions here.

\appendix
% \section{Additional details on the boundary eigenvalue problem} 
% \subsection{Equivalent first order system}
% The boundary eigenvalue operator function $\tmatT_\bgamma$ can be mapped to an equivalent first-order system on the spaces:
% \begin{align}
%     \hbbH &:= \lp \mathrm{L}^2(\alpha, H)\rp^2\times \lp \mathrm{L}^2(H,\theta)\rp^2 \times \mathrm{L}^2(0, \tau_0) \\
%     \hbbW &:= \lp \mathrm{W}^{2,2}(\alpha, H)\rp^2\times \lp \mathrm{W}^{2,2}(H,\theta)\rp^2 \times \mathrm{W}^{2,2}(0, \tau_0)
% \end{align}
% $\widehat \matT_\bgamma \in \hol{\bbC}{\linear{\hbbW}{\hbbH\times\bbC^5}}$ 
% where $\hbbH:= \lp \mathrm{L}^2(\alpha, H)\rp^2\times \lp \mathrm{L}^2(H,\theta)\rp^2 \times \mathrm{L}^2(0, \tau_0)$ and $\hbbW:= \lp \mathrm{W}^{2,2}(\alpha, H)\rp^2\times \lp \mathrm{W}^{2,2}(H,\theta)\rp^2 \times \mathrm{W}^{2,2}(0, \tau_0)$

% MM: LIKELY TO BE REMOVED\dots

\section{Reduced adjoint boundary eigenvalue problem}
\label{sec:supp-adjoint}

\subsection{The adjoint problem}
    To study the properties of the spectrum of the evolution operator, it is convenient to consider the adjoint boundary eigenvalue problem, as the domain has no derivative discontinuity at $H$ and the boundary conditions can be reduced. Indeed , inspecting 2nd and 3rd boundary condition in equation \eqref{eq:adj-bc} we observe that $\matB^*\mathbf{f} = 0 \Rightarrow f\in \mathrm{W}^{2,\pzq}(\alpha, \theta)$.
    We can then rewrite the domain of the adjoint operator as 
\begin{align}
    \bbD^* := \{\mathbf{f} \in \bbV^{\pzq} | \mathcal{B}^\oplus \mathbf{f} = 0\}
\end{align}
where $\bbV^{\pzq} := \mathrm{W}^{2,\pzq}(\alpha, \theta)\times \mathrm{W}^{1,\pzq}(0, \tau_0)$ and $\mathcal{B}^\oplus: \bbV^{\pzq} \to \bbC^3$ is the reduced adjoint boundary operator, defined as:
\begin{align}
    \mathcal{B}^\oplus: \bbV^{\pzq} &\to \bbC^3 \\
    \mathbf{f} = (f, f^\sr) &\mapsto
\begin{bmatrix}
    f'(\alpha)\\
    f(H) - f^\sr(\tau_0) \\
    f(\theta) - f^\sr(0)
\end{bmatrix}    
\end{align} 
    We then define the reduced adjoint boundary eigenvalue operator as
    the holomorphic function that maps the complex parameter $\lambda\in\bbC$ to the operator: 
\begin{align}
    \tmatT_\bgamma^\oplus(\lambda)=(\tmatT^{\mathtt{D}\oplus },\tmatT^{\mathtt{R}\oplus}):\bbV &\to \bbL^\pzq \times \bbC^3\\
    \mathbf{f}=\lp f, f^\sr\rp &\mapsto \lp \lp\matL^*_\bgamma f, \partial_\tau f^\sr\rp - \lambda \mathbf{f}, \matB^\oplus \mathbf{f}\rp
\end{align}
$\tmatT^{\oplus D}_\bgamma(\lambda): \bbW^\pzp \to \bbL^\pzp$ has a right inverse $\matQ_\bgamma(\lambda):\bbL^\pzp \to \bbW^\pzp$, holomorphic in $\lambda$, given by the variation of parameters formula: 
\begin{align}\notag
    \Dp{\pzq} \ni\mathbf{g} = \begin{bmatrix}
        g\\ g^\sr
    \end{bmatrix} \mapsto &\matQ^\oplus_\bgamma(\lambda) \mathbf{g} = \\ 
    &=\begin{bmatrix}
        v \mapsto\! \int_\alpha^vg(u)\cf{[\psi_1(v, \lambda, \bgamma)f_2(u, \lambda, \bgamma) -  \psi_2(v, \lambda, \bgamma)f_1(u, \lambda, \bgamma)]} du \vspace{0.5em}\\
        \tau \mapsto\! e^{\tau\lambda}\int_0^\tau e^{-z\lambda}g^\sr(z)dz
    \end{bmatrix}
\end{align}
        
As before, we define the reduced adjoint fundamental matrix function as the holomorphic function $\matZ^\oplus_\bgamma\in \hol{\bbC}{{\mathfrak{L}\lp\bbC^3, \bbV\rp}}$  that parametrizes the kernel of $\tmatT_\bgamma^{D\oplus}$:
\begin{align}
    \matZ^\oplus_\bgamma(\lambda): \bbC^3 &\to \bbV\\
    \lp a, b, a^\sr\rp &\mapsto 
    \begin{bmatrix}
        a \psi_1 + b \psi_2(\cdot, \lambda, \bgamma)\\
        a^\sr \exp(\cdot\lambda)
    \end{bmatrix} \in \bbV
\end{align}
Where $\psi_1(\cdot, \lambda, \bgamma) , \psi_2(\cdot, \lambda, \bgamma) \in \mathrm{W}^{2,\pzp}\lp\alpha, \theta\rp$ are the two independent fundamental solutions of the equation $\matL^*_\bgamma \psi =\lambda \psi$ defined in Section \ref{sec:spectrum}.  
% such that:
% \begin{align}\label{eq:fund-sol}
%     &\matL^*_\bgamma \psi_i(v, \lambda, \bgamma) = \lambda \psi_i(v, \lambda, \bgamma)\quad \forall v\in (\alpha, \theta)\\
%     & \psi'_i(\alpha, \lambda, \bgamma) = -\delta_{i,2} \quad \psi_i(\alpha, \lambda, \bgamma) = \delta_{i,1}\quad \forall i \in \{1,2\}\ , \forall \lambda \in \bbC
% \end{align}
% such that $\psi_i = w_\bgamma^{-1}\cdot f_i$ \luca{insert citation}.
The spectrum structure can be completely determined by studying the $3\times 3$ reduced adjoint characteristic matrix $M^\oplus_\bgamma\in\hol{\bbC}{\linear{\bbC^3}{\bbC^3}}$, defined as:
\begin{align}
    M^\oplus_\bgamma(\lambda) = \matB^\oplus \matZ^\oplus_\bgamma(\lambda) =
    \begin{bmatrix}
        0 & -\DD^{-1} & 0 \\
        \psi_1(H,\lambda, \bgamma) & \psi_2(H,\lambda, \bgamma)& -\exp(\lambda\tau_0) \\
        \psi_1(\theta,\lambda, \bgamma) & \psi_2(\theta,\lambda, \bgamma)& -1      \end{bmatrix}
    \quad \forall \lambda \in \bbC
\end{align}
We can then rewrite the characteristic equation as:
\begin{align*}
    0 = \Delta_\bgamma(\lambda) = \cf{\DD \exp(-\lambda\tau_0)\det \lp M^\oplus_\bgamma(\lambda)\rp = \psi_1(\theta, \bgamma, \lambda) - \psi_1(H, \bgamma, \lambda)\exp(-\tau_0\lambda)}
\end{align*}

\subsection{Eigenfuncions and associated functions}

\begin{proof}[Proof of Lemma~\ref{lem:root-functions}]
By Theorem~\ref{thm:spectrum-prop}, every eigenvalue has geometric multiplicity one.
Hence, by Theorem~1.5.9 and Proposition~3.4.7 in \cite{mennicken2003non}, for each
$\lambda_n\in\sigma(\bgamma)$ there exist holomorphic root functions
$s\mapsto\bphi(s,\bgamma)$ and $s\mapsto\bpsi(s,\bgamma)$, defined in a neighborhood of
$\lambda_n$, such that the principal part of the resolvent is described by
\eqref{eq:principal-part-resolvent}. By \cite[Lemma~1.11.2]{mennicken2003non}, these
root functions can be constructed explicitly from a characteristic matrix and its adjoint.
In the present setting it is convenient to use the reduced adjoint characteristic matrix
$M^\oplus_\bgamma$.

% Recall that, with the normalization $\psi_i=w_\bgamma^{-1}f_i$, one has
% \[
% \psi_1'(\alpha,s,\bgamma)=0,
% \qquad
% \psi_2'(\alpha,s,\bgamma)=-D^{-1},
% \]
% and therefore
% \begin{align}
%     M^\oplus_\bgamma(s)
%     =
%     \begin{bmatrix}
%         0 & -D^{-1} & 0\\
%         \psi_1(H,s,\bgamma) & \psi_2(H,s,\bgamma) & -e^{s\tau_0}\\
%         \psi_1(\theta,s,\bgamma) & \psi_2(\theta,s,\bgamma) & -1
%     \end{bmatrix}.
% \end{align}
% Its determinant is
% \begin{align}
%     \det M^\oplus_\bgamma(s)
%     =
%     D^{-1}\Big(e^{s\tau_0}\psi_1(\theta,s,\bgamma)-\psi_1(H,s,\bgamma)\Big)
%     =
%     -D^{-1}e^{s\tau_0}\Delta_\bgamma(s),
% \end{align}
% where
% \begin{align}
%     \Delta_\bgamma(s):=e^{-s\tau_0}\psi_1(H,s,\bgamma)-\psi_1(\theta,s,\bgamma).
% \end{align}
% Hence $\det M^\oplus_\bgamma$ and $\Delta_\bgamma$ have the same zeros, with the same
% multiplicities.

We first construct holomorphic root functions for $M^\oplus_\bgamma$ and its adjoint.
A convenient right kernel vector is
\begin{align}
    \mathbf c(s,\bgamma)=
    \begin{bmatrix}
        1\\
        0\\
        \psi_1(\theta,s,\bgamma)
    \end{bmatrix},
\end{align}
since
\begin{align}
    M^\oplus_\bgamma(s)\mathbf c(s,\bgamma)
    =
    \begin{bmatrix}
        0\\
        \psi_1(H,s,\bgamma)-e^{s\tau_0}\psi_1(\theta,s,\bgamma)\\
        0
    \end{bmatrix}
    =
    \cf{-e^{s\tau_0}\Delta_\bgamma(s)}
    \begin{bmatrix}
        0\\
        1\\
        0
    \end{bmatrix}.
\end{align}
A corresponding left kernel vector can be obtained from the cofactors of
$M^\oplus_\bgamma(s)$; one convenient choice is
\begin{align}
    \mathbf d(s,\bgamma)=
    \begin{bmatrix}
        e^{-s\tau_0}\psi_2(H,s,\bgamma)-\psi_2(\theta,s,\bgamma)\\
        \cf{D^{-1}e^{-s\tau_0}}\\
        -D^{-1}
    \end{bmatrix}.
\end{align}
Indeed, a direct computation shows that
\begin{align}
    \mathbf d(s,\bgamma)^*M^\oplus_\bgamma(s)
    =
    \mathbf 0
    \qquad\text{whenever}\qquad
    \Delta_\bgamma(s)=0.
\end{align}

We then define
\begin{align}
    \bpsi(s,\bgamma)
    &=
    \matZ^\oplus_\bgamma(s)\mathbf c(s,\bgamma)
    =
    \begin{bmatrix}
        \psi_1(\,\cdot\,,s,\bgamma)\\
        \psi_1(\theta,s,\bgamma)e^{\,\cdot\, s}
    \end{bmatrix},
\end{align}
and
\begin{align}
    \bphi(s,\bgamma)\notag
    &=
    \big(B^\oplus \matQ^\oplus_\bgamma(s)\big)^*\mathbf d(s,\bgamma) \\
    &=
    \begin{bmatrix}
        f_1^-(\,\cdot\,,s,\bgamma)\Big(e^{-s\tau_0}\psi_2(H,s,\bgamma)-\psi_2(\theta,s,\bgamma)\Big)
        \cf{+f_2^-(\,\cdot\,,s,\bgamma)\Delta_\bgamma(s)}\\
        f_2^+(\,\cdot\,,s,\bgamma)\psi_1(\theta,s,\bgamma)
        -f_1^+(\,\cdot\,,s,\bgamma)\psi_2(\theta,s,\bgamma)\\
        e^{-\,\cdot\, s}
    \end{bmatrix}.
\end{align}
By construction $\bphi(\,\cdot\,,\bgamma)$ and $\bpsi(\,\cdot\,,\bgamma)$ are holomorphic root functions, whose Taylor coefficients at $\lambda_n$ furnish the associated vectors. \cf{A direct computation gives $\mathbf d(s,\bgamma)^{*}M^\oplus_\bgamma(s)\mathbf c(s,\bgamma)=-\DD^{-1}\Delta_\bgamma(s)$; since $\matZ^\oplus_\bgamma$ and $(\matB^\oplus\matQ^\oplus_\bgamma)^{*}$ contribute a compensating factor $-\DD$, one has
$\langle\bpsi(\lambda_n,\bgamma),\bphi(\lambda_n,\bgamma)\rangle=\Delta_\bgamma'(\lambda_n)$.} The eigenvalue $\lambda_n$ having geometric multiplicity one, a single constant normalizes the \cb{eigenfunctions}, \cb{namely}
\begin{align}
    Z_n=\lim_{s\to\lambda_n}\frac{\Delta_\bgamma(s)}{(s-\lambda_n)^{m_n}}\in\bbC\setminus\{0\},
\end{align}
the leading coefficient of $\Delta_\bgamma$ at $\lambda_n$ (so $Z_n=\Delta_\bgamma'(\lambda_n)$ for $\lambda_n$ simple), the residual scale being fixed by $N[\bphi]=1$; \eqref{eq:principal-part-resolvent} then follows from \cite[Lemma~1.11.2]{mennicken2003non}.
\cb{For $m_n\ge2$ this single constant no longer normalizes the whole canonical system: a
defective eigenvalue is a multiple root of $\Delta_\bgamma$, so $\Delta_\bgamma'(\lambda_n)=0$ and
hence $\langle\bpsi_n,\bphi_n\rangle=0$; no rescaling of the eigenfunctions can make it one. The
remaining associated vectors are then fixed not by a scalar but by the chain biorthogonality of
\cite[Thm.~1.5.4 and Cor.~1.5.6]{mennicken2003non}, which is what
\eqref{eq:assoc-vectors-suppl} below implements.}
\end{proof}

\paragraph{Eigenfunctions and normalization.}
The functions $\bphi(s,\bgamma)$ and $\bpsi(s,\bgamma)$ form, after normalization at each
eigenvalue $\lambda_n$, \mmtext{a CSRF} for
$\tmatT_\bgamma$ and its adjoint. For $n>0$ we define the eigenfunctions of
$\matT_\bgamma$ and $\matT_\bgamma^*$ corresponding to $\lambda_n$ by
\[
\bphi_n(\bgamma):=\bphi(\lambda_n,\bgamma),
\qquad
\bpsi_n(\bgamma):=Z_n^{-1}\bpsi(\lambda_n,\bgamma).
\]
For the stationary eigenvalue $n=0$ we instead set
\[
\bphi_0(\bgamma):=Z_0^{-1}\bphi(0,\bgamma),
\qquad
\bpsi_0(\bgamma):=\bpsi(0,\bgamma),
\]
so that $\|\bphi_0\|_{\bbL^1}=1$ and $\bphi_0$ coincides with the stationary probability
distribution.

\paragraph{Associated vectors and Laurent expansion.}
Let $\lambda_n$ be an eigenvalue of algebraic multiplicity $m_n$. The associated vectors
are obtained from the Taylor expansion of the holomorphic root functions at $s=\lambda_n$:
\begin{align}
    \bphi_n^{(k)}(\bgamma)
    &:=
    \frac{1}{k!}\,\partial_s^k\bphi(s,\bgamma)\big|_{s=\lambda_n}
    \in \bbD_\bgamma,
    \\
    \bpsi_n^{(k)}(\bgamma)
    &:=
    \cb{\frac{1}{k!}\,\partial_s^k
    \!\left[\frac{(s-\lambda_n)^{m_n}\,\bpsi(s,\bgamma)}{\Delta_\bgamma(s)}\right]_{s=\lambda_n}}
    \in \bbD^*,
    \qquad
    k=0,1,\dots,m_n-1.
    \label{eq:assoc-vectors-suppl}
\end{align}
These vectors form Jordan chains for $\tmatT_\bgamma$ and its adjoint:
\begin{align}
    (\matT_\bgamma-\lambda_n)\bphi_n^{(k+1)}=\bphi_n^{(k)},
    \qquad
    (\matT_\bgamma^*-\lambda_n)\bpsi_n^{(k+1)}=\bpsi_n^{(k)},
    \qquad
    0\le k\le m_n-2.
\end{align}
Expanding $\bphi(\,\cdot\,,\bgamma)$, $\bpsi(\,\cdot\,,\bgamma)$, and
$\Delta_\bgamma$ at $s=\lambda_n$, and inserting the resulting series into
\eqref{eq:principal-part-resolvent}, yields the singular Laurent part of the resolvent:
\begin{equation}\label{eq:pp_expansion}
    \Big[(s-\matT_\bgamma)^{-1}\Big]_{\mathrm{sing},\,\lambda_n}
    =
    \sum_{r=1}^{m_n}(s-\lambda_n)^{-r}
    \sum_{k=0}^{m_n-r}
    \bphi_n^{(k)}\otimes\bpsi_n^{(m_n-r-k)}.
\end{equation}

\paragraph{Biorthogonality of the CSEAV.}
Let $\lambda_i,\lambda_j\in\sigma(\bgamma)$ have algebraic multiplicities $m_i$ and $m_j$,
respectively. Then the associated vectors form a biorthogonal canonical system of
eigenvectors and associated vectors (CSEAV), namely
\begin{equation}\label{eq:CSEAV_biorth}
    \big\langle \bpsi_i^{(l)},\,\bphi_j^{(m_j-1-k)}\big\rangle
    =
    \delta_{ij}\delta_{lk},
    \qquad
    0\le l<m_i,
    \quad
    0\le k<m_j.
\end{equation}
In particular, for each fixed eigenvalue $\lambda_n$, the left chain
$\{\bpsi_n^{(0)},\dots,\bpsi_n^{(m_n-1)}\}$ pairs with the right chain
$\{\bphi_n^{(m_n-1)},\dots,\bphi_n^{(0)}\}$ in reversed order.

\subsubsection{Flux properties}

Finally, we characterize the firing rate associated with the spectral modes. This is
physically relevant because the total population firing rate is obtained as a superposition
of the modal firing rates.

\begin{lemma}[Flux normalization]\label{lem:flux-normalization}
For every eigenvalue $\lambda_n\in\sigma(\bgamma)$, the eigenfunction $\bphi_n^{(0)}$
has unit firing rate, while all associated functions $\bphi_n^{(k)}$ with $k\ge1$ have
vanishing firing rate: \mmnote{I changed $\mathcal N$ into $N$ according to the main text notation}
\[
%\mathcal 
N[\bphi_n^{(0)}]=1,
\qquad
%\mathcal 
N[\bphi_n^{(k)}]=0,
\quad k=1,\dots,m_n-1.
\]
Equivalently,
\[
\matS_\bgamma \phi_n^{(0)}(\theta)=1,
\qquad
\matS_\bgamma \phi_n^{(k)}(\theta)=0,
\quad k=1,\dots,m_n-1.
\]
\end{lemma}

\begin{proof}
Recall that the firing-rate functional is given by
\[
%\mathcal 
N[\mathbf p]=p^{\sr}(0),
\]
and, by the boundary condition at $\theta$, it coincides with the threshold flux
$\matS_\bgamma p^+(\theta)$.

For the generating root function $\bphi(s,\bgamma)$, the refractory component is
\[
\phi^{\sr}(\tau,s,\bgamma)=e^{-\tau s}.
\]
Hence
\[
%\mathcal 
N[\bphi(s,\bgamma)]=\phi^{\sr}(0,s,\bgamma)=1
\qquad \forall s.
\]
Evaluating at $s=\lambda_n$ gives
\[
%\mathcal 
N[\bphi_n^{(0)}]=1.
\]

For $k\ge1$, differentiating with respect to $s$ yields
\begin{align}
    %\mathcal 
    N[\bphi_n^{(k)}]
    &=
    \frac{1}{k!}\partial_s^k
    %\mathcal 
    N[\bphi(s,\bgamma)]\big|_{s=\lambda_n}
    =
    \frac{1}{k!}\partial_s^k(1)\big|_{s=\lambda_n}
    =
    0.
\end{align}
% Equivalently, since $\phi^\sr(\tau,s,\bgamma)=e^{-\tau s}$,
% \[
% \mathcal N[\bphi_n^{(k)}]
% =
% \frac{1}{k!}\partial_s^k\!\left[e^{-\tau s}\right]_{\tau=0,\,s=\lambda_n}
% =
% \frac{1}{k!}\left[(-\tau)^k e^{-\tau s}\right]_{\tau=0,\,s=\lambda_n}
% =
% 0,
% \qquad k\ge1.
% \]
% The equivalent flux statement follows from the boundary identity
% $\mathcal N[\mathbf p]=\matS_\bgamma p^+(\theta)$ on $\bbD_\bgamma$.
\end{proof}

\section{Further details on dissipativity} 
\label{app-dissipative}

Given $\mathbf{p} = (p, p^r) \in \bbL^\pzp \subseteq \bbL^1, \pzp \in [1, \infty]$ we define the sign of $\mathbf{p}$ as $\mathrm{sign}(\mathbf{p}) = (\mathrm{sign}(p), \mathrm{sign}(p^r)) \in \bbL^\infty$ where: 
\begin{align}\notag
    \sign(p)(v) = 
    \begin{cases}
        \frac{\overline{p(v)}}{|p(v)|} & p(v) \neq 0 \\
        0 & p(v) = 0
    \end{cases} \qquad 
    \sign(p^r)(\tau) = 
    \begin{cases}
        \frac{\overline{p^r\!(\tau)}}{|p^r(\tau)|} & p^r\!(\tau) \neq 0 \\
        0 & p^r(\tau) = 0
    \end{cases} 
    \begin{aligned}
        \quad v&\in(\alpha, \theta)\\ \tau &\in (0, \tau_0)
    \end{aligned}
\end{align}
such that $\langle \mathrm{sign}(\mathbf{p}),  \mathbf{p}\rangle = \|\bp\|_{\bbL^1}$. Then $ \|\bp\|_{\bbL^1} \sign \bp \in \bbL^\infty$ is a element of the duality set of $\bp \in \bbL^1$: 
\begin{equation}
    \|\bp\|_{\bbL^1}\ \sign \bp  \in  J(p) = \Set{\mathbf{f} \in \bbL^\infty | \langle \mathbf{f}, \mathbf{p} \rangle =  \|\bp\|_{\bbL^1}^2}
\end{equation}
If we rewrite $\bp = \lp p, p^r\rp$ is polar cordinates: 
\begin{align}
    p(v) = h(v)e^{i\beta(v)} \quad p^r(\tau) = h^r(\tau)e^{i\beta^r(\tau)}
\end{align}
Where $\bbR \ni h(v), h^r(\tau) \ge 0$ and $\beta(v), \beta^r(\tau) \in \bbR$. The sign is given by:
\begin{align}\notag
    \sign(p)(v) = 
    \begin{cases}
        e^{-i\beta(v)} & h(v) > 0 \\
        0 & h(v) = 0
    \end{cases} \quad 
    \sign(p^r)(\tau) = 
    \begin{cases}
        e^{-i\beta^r(\tau)} & h^r(\tau) > 0 \\
        0 & h(\tau) = 0
    \end{cases} 
    \begin{aligned}
        \quad v&\in(\alpha, \theta)\\ \tau &\in (0, \tau_0)
    \end{aligned}
\end{align}

\begin{proof}[Proof of Lemma \ref{lemma-dissipative}]
    As $\bp \in \prescript{1}{}\bbD_\bgamma$ is a continuous function the set $\{v\in (\alpha, \theta): h(v)>0\}$ is open and can be written as a (at most) countable union of intervals 
    $\{v\in (\alpha, \theta): h(v)>0\} = \cup_{n} (a_n, b_n)$. 
    Integrating by parts and imposing boundary conditions we have: 
    \begin{align}\notag
         &Re \lp\int_\alpha^\theta \sign \lp p(v) \rp \matL_\bgamma p(v) dv \cf{-} \int_0^{\tau_0} \sign \lp p^r(\tau)\rp \partial_\tau p^r(\tau) d\tau \rp = \\
         =&\label{eq:diss-term} -\DD\sum_n \int_{a_n}^{b_n} h(v)\cf{(\beta'(v))^2} dv \cf{+ \DD\sum_n \lp h'(b_n^-) - h'(a_n^+)\rp} + \\   &+ h^r(0)(\cos(\beta(\theta) - \beta^r(0)) - 1) + h^r(\tau_0) \notag (\cos(\beta(H)  - \beta^r(\tau_0)) - 1) \le 0\end{align}
 notice that since $h(v) \ge 0\ \forall v \in (\alpha, \theta)$ and $h(a_n) = h(b_n) = 0$, we have $h'(b_n^-)\le 0,  h'(a_n^+)\ge 0$. Dissipativity follows from \cite[Proposition 3.23] {Engel1999OneparameterSF}
\end{proof}

    and $0=\lambda_0\in \sigma(\bgamma)$ is the only eigenvalue with non-negative real part.
    \begin{proof}[Proof of Proposition \ref{prop:eig-inequality}]
        Let $\lambda \in \sigma(\bgamma)$, with $\bphi(\lambda, \bgamma)\in \prescript{\pzp}{}{\bbD_\bgamma}\subset  \prescript{1}{}{\bbD_\bgamma}$, then by \eqref{eq:diss-ineq}:
        \begin{equation}
            Re(\lambda) \|\bphi(\lambda, \bgamma)\|_{\bbL^1} \le 0
        \end{equation}
        This proves \eqref{eq:sp-lefthalf}. Suppose now that there exists $i\omega \in \sigma(\bgamma)$ with $\omega \in \bbR$. Up to a complex scalar factor, the corresponding eigenfunction will have the form $\bphi(i\omega, \bgamma) = (h_\omega(v)exp(i\beta_\omega(v)), \exp(i\tau\omega))$. We then have: 
        \begin{equation}
            Re \langle \sign \bphi(i\omega, \bgamma), \matT_\bgamma \bphi(i\omega, \bgamma)\cf{\rangle} = Re(i\omega\|\bphi(i\omega, \bgamma)\|_{\bbL^1}) = 0
        \end{equation}
        This holds if and only if all the negative terms in \eqref{eq:diss-term} vanish. The reader can verify that this occurs precisely when $\beta_w(v) \equiv 0 \land h_w(v) > 0,\ \forall v\in (\alpha, \theta)$ and $\omega$ is of the form $\omega = n2\pi/\tau_0$ for some $n \in \mathbb{Z}$. Then as $0\in \sigma(\bgamma) \Rightarrow \bpsi(0, \bgamma)\in \bbD^*$ we have:
        \begin{align}
            0 = \langle \bpsi(0, \bgamma), \matT_\bgamma \bphi(iw, \bgamma) \rangle = in\frac{2\pi}{\tau_0} \lp \int_\alpha^\theta h_\omega(v)dv + \int_0^{\tau_0} cos\lp n\frac{2\pi \tau}{\tau_0}\rp d \tau \rp 
        \end{align}
        this implies $n = 0 \Rightarrow \omega = 0$.
    \end{proof}
%       
% \begin{lemma}
% Test Lemma.
% \end{lemma}

% \section*{Acknowledgments}
% We would like to acknowledge the assistance of volunteers in putting together this example manuscript and supplement.

\bibliographystyle{plain}
\bibliography{references}

@article{Amit1997,
author = {Amit, Daniel J and Brunel, Nicolas},
issn = {1047-3211},
journal = {Cereb. Cortex},
number = {3},
pages = {237--52},
pmid = {9143444},
title = {{Model of global spontaneous activity and local structured activity during delay periods in the cerebral cortex.}},
volume = {7},
year = {1997}
}

@article{Berry1998,
author = {{Berry II}, Michael J and Meister, Markus},
doi = {10.1523/JNEUROSCI.18-06-02200.1998},
journal = {J. Neurosci.},
number = {6},
pages = {2200--2211},
pmid = {9482804},
title = {{Refractoriness and neural precision}},
volume = {18},
year = {1998}
}

@article{Schuecker2015,
author = {Schuecker, Jannis and Diesmann, Markus and Helias, Moritz},
doi = {10.1103/PhysRevE.92.052119},
journal = {Phys. Rev. E},
number = {5},
pages = {052119},
pmid = {26651659},
title = {{Modulated escape from a metastable state driven by colored noise.}},
volume = {92},
year = {2015}
}

@article{Lindner2001,
author = {Lindner, Benjamin and Schimansky-Geier, Lutz},
doi = {10.1103/PhysRevLett.86.2934},
journal = {Phys. Rev. Lett.},
number = {14},
pages = {2934--2937},
title = {{Transmission of noise coded versus additive signals through a neuronal ensemble}},
volume = {86},
year = {2001}
}

@article{Lindner2004,
author = {Lindner, Benjamin and Garc{\'{i}}a-Ojalvo, Jordi and Neiman, Alexander B and Schimansky-Geier, Lutz},
doi = {10.1016/j.physrep.2003.10.015},
journal = {Phys. Rep.},
number = {6},
pages = {321--424},
title = {{Effects of noise in excitable systems}},
volume = {392},
year = {2004}
}

@article{Mattia2004,
author = {Mattia, Maurizio and {Del Giudice}, Paolo},
doi = {10.1103/PhysRevE.70.052903},
journal = {Phys. Rev. E},
number = {5 Pt 1},
pages = {052903},
pmid = {15600672},
title = {{Finite-size dynamics of inhibitory and excitatory interacting spiking neurons.}},
volume = {70},
year = {2004}
}

@article{Mattia2021,
author = {Mattia, Maurizio and Vinci, Gianni V},
doi = {10.5281/zenodo.5518215},
journal = {Zenodo},
pages = {5518215},
title = {{Low dimensional dynamics of spiking neuron networks}},
year = {2021}
}

@article{Treves1993,
author = {Treves, Alessandro},
journal = {Network},
number = {3},
pages = {259--84},
title = {{Mean-field analysis of neuronal spike dynamics}},
volume = {4},
year = {1993}
}

@article{Richardson2004,
author = {Richardson, Magnus J E},
doi = {10.1103/PhysRevE.69.051918},
journal = {Phys. Rev. E},
number = {5 Pt 1},
pages = {051918},
pmid = {15244858},
title = {{Effects of synaptic conductance on the voltage distribution and firing rate of spiking neurons.}},
volume = {69},
year = {2004}
}

@article{Richardson2007,
author = {Richardson, Magnus J E},
doi = {10.1103/PhysRevE.76.021919},
journal = {Phys. Rev. E},
number = {2 Pt 1},
pages = {021919},
pmid = {17930077},
title = {{Firing-rate response of linear and nonlinear integrate-and-fire neurons to modulated current-based and conductance-based synaptic drive}},
volume = {76},
year = {2007}
}

@article{Ricciardi1988,
author = {Ricciardi, Luigi M and Sato, Shunsuke},
doi = {10.2307/3214232},
journal = {J. Appl. Prob.},
number = {1},
pages = {43--57},
title = {{First-passage-time density and moments of the Ornstein-Uhlenbeck process}},
volume = {25},
year = {1988}
}

@article{Engel1999OneparameterSF,
  title={One-parameter semigroups for linear evolution equations},
  author={Klaus-Jochen Engel and Rainer Nagel},
  journal={Semigroup Forum},
  year={1999},
  volume={63},
  pages={278-280},
  url={https://api.semanticscholar.org/CorpusID:117061340}
}

@book{Batkai2005,
  title={Semigroups for Delay Equations},
  author={B{\'a}tkai, Andr{\'a}s and Piazzera, Susanna},
  series={Research Notes in Mathematics},
  volume={10},
  publisher={A K Peters},
  isbn={9781568812434},
  address={Wellesley, MA},
  year={2005}
}

@article{ben2009ergodic,
  title={Ergodic behavior of diffusions with random jumps from the boundary},
  author={Ben-Ari, Iddo and Pinsky, Ross G},
  journal={Stochastic processes and their applications},
  volume={119},
  number={3},
  pages={864--881},
  year={2009},
  publisher={Elsevier}
}

@article{Augustin2017,
  title={Low-dimensional spike rate models derived from networks of adaptive integrate-and-fire neurons: comparison and implementation},
  author={Augustin, Moritz and Ladenbauer, Josef and Baumann, Fabian and Obermayer, Klaus},
  journal={PLoS Comput. Biol.},
  volume={13},
  number={6},
  pages={e1005545},
  year={2017}
}

@article{Leung2022,
author = {Leung, Yuk-J},
doi = {10.1002/mma.8493},
journal = {Math. Meth. Appl. Sci.},
number = {17},
pages = {13051--13062},
title = {{One dimensional Brownian motion with holding and jumping boundary}},
volume = {47},
year = {2024}
}

@article{Peng2012DiffusionsWH,
  title={Diffusions with holding and jumping boundary},
  author={Jun Peng and Wenbo V. Li},
  journal={Science China Mathematics},
  year={2012},
  volume={56},
  pages={161 - 176},
  url={https://api.semanticscholar.org/CorpusID:255161142}
}

@article{Grigorescu2002,
author = {Grigorescu, Ilie and Kang, Min},
doi = {10.1023/A:1016232201962},
journal = {J. Theor. Prob.},
number = {3},
pages = {817--844},
title = {{Brownian motion on the figure eight}},
volume = {15},
year = {2002}
}

@Misc{amsmath,
  author =	 {{American Mathematical Society}},
  title =	 {User's Guide for the \texttt{amsmath} Package
                  (Version 2.0)},
  url =		 {ftp://ftp.ams.org/pub/tex/doc/amsmath/amsldoc.pdf},
  urldate =	 {2015-07-30},
  year =	 2002}

@article{Abbott1993,
author = {Abbott, Larry F and van Vreeswijk, Carl},
journal = {Physical Review E},
number = {2},
pages = {1483--1490},
pmid = {9960738},
title = {{Asynchronous states in networks of pulse-coupled oscillators.}},
url = {http://link.aps.org/doi/10.1103/PhysRevE.48.1483 http://www.ncbi.nlm.nih.gov/pubmed/9960738},
volume = {48},
year = {1993}
}

@article{omurtag2000dynamics,
  title={Dynamics of neuronal populations: The equilibrium solution},
  author={Omurtag, Ahmet and Knight, Bruce W and Sirovich, Lawrence},
  journal={SIAM Journal on Applied Mathematics},
  volume={60},
  number={6},
  pages={2009--2028},
  year={2000},
  publisher={SIAM}
}

@inproceedings{Knight1996,
address = {Cite Scientifique, Lille, France},
author = {Knight, Bruce W and Manin, Dimitri and Sirovich, Lawrence},
booktitle = {Symposium on Robotics and Cybernetics: Computational Engineering in Systems Applications},
editor = {Gerf, E.C.},
pages = {1--5},
publisher = {Cite Scientifique},
title = {{Dynamical models of interacting neuron populations in visual cortex}},
year = {1996}
}

@article{knight2000dynamics,
  title={Dynamics of encoding in neuron populations: some general mathematical features},
  author={Knight, Bruce W},
  journal={Neural Computation},
  volume={12},
  number={3},
  pages={473--518},
  year={2000},
  publisher={MIT Press}
}

@book{tuckwell1988introduction,
  title={Introduction to theoretical neurobiology: linear cable theory and dendritic structure},
  author={Tuckwell, Henry Clavering},
  volume={1},
  year={1988},
  publisher={Cambridge University Press}
}

@article{Mattia2002,
  title={Population dynamics of interacting spiking neurons},
  author={Mattia, Maurizio and Del Giudice, Paolo},
  journal={Phys. Rev. E},
  volume={66},
  number={5},
  pages={051917},
  year={2002},
  publisher={APS}
}

@article{brunel1999fast,
  title={Fast global oscillations in networks of integrate-and-fire neurons with low firing rates},
  author={Brunel, Nicolas and Hakim, Vincent},
  journal={Neural computation},
  volume={11},
  number={7},
  pages={1621--1671},
  year={1999},
  publisher={MIT Press One Rogers Street, Cambridge, MA 02142-1209, USA journals-info~...}
}

@article{brunel2000dynamics,
  title={Dynamics of sparsely connected networks of excitatory and inhibitory spiking neurons},
  author={Brunel, Nicolas},
  journal={Journal of computational neuroscience},
  volume={8},
  number={3},
  pages={183--208},
  year={2000},
  publisher={Springer}
}

@article{caceres2011analysis,
  title={Analysis of nonlinear noisy integrate \& fire neuron models: blow-up and steady states},
  author={C{\'a}ceres, Mar{\'\i}a J and Carrillo, Jos{\'e} A and Perthame, Beno{\^{i}}t},
  journal={The Journal of Mathematical Neuroscience},
  volume={1},
  pages={1--33},
  year={2011},
  publisher={Springer}
}

@article{Fusi1999,
author = {Fusi, Stefano and Mattia, Maurizio},
doi = {10.1162/089976699300016601},
journal = {Neural Computation},
number = {3},
pages = {633--52},
pmid = {10085424},
title = {{Collective behavior of networks with linear (VLSI) integrate-and-fire neurons.}},
url = {http://www.ncbi.nlm.nih.gov/pubmed/10085424},
volume = {11},
year = {1999}
}

@article{Brinkman2022,
author = {Brinkman, Braden A W and Yan, Han and Maffei, Arianna and Park, Il Memming and Fontanini, Alfredo and Wang, Jin and {La Camera}, Giancarlo},
doi = {10.1063/5.0062603},
journal = {Applied Physics Reviews},
number = {1},
pages = {011313},
title = {{Metastable dynamics of neural circuits and networks}},
volume = {9},
year = {2022}
}

@article{liu2022rigorous,
  title={Rigorous Justification of the Fokker--Planck Equations of Neural Networks Based on an Iteration Perspective},
  author={Liu, Jian-Guo and Wang, Ziheng and Zhang, Yuan and Zhou, Zhennan},
  journal={SIAM Journal on Mathematical Analysis},
  volume={54},
  number={1},
  pages={1270--1312},
  year={2022},
  publisher={SIAM}
}

@article{rdm,
title = {Mind the last spike --- firing rate models for mesoscopic populations of spiking neurons},
journal = {Curr. Opin. Neurobiol.},
volume = {58},
pages = {155-166},
year = {2019},
doi = {10.1016/j.conb.2019.08.003},
author = {Tilo Schwalger and Anton V Chizhov}
}

@article{srm,
  title = {Time structure of the activity in neural network models},
  author = {Gerstner, Wulfram},
  journal = {Phys. Rev. E},
  volume = {51},
  issue = {1},
  pages = {738--758},
  numpages = {0},
  year = {1995},
  month = {Jan},
  publisher = {American Physical Society},
  doi = {10.1103/PhysRevE.51.738},
}

@book{Gerstner2014,
author = {Gerstner, Wulfram and Kistler, Werner M. and Naud, Richard and Paninski, Liam},
booktitle = {Neuronal Dynamics: From Single Neurons to Networks and Models of Cognition},
doi = {10.1017/CBO9781107447615},
isbn = {9781107635197},
pages = {1--577},
publisher = {Cambridge University Press},
title = {{Neuronal dynamics: From single neurons to networks and models of cognition}},
url = {www.cambridge.org/9781107635197},
year = {2014}
}

@book{mennicken2003non,
  title={Non-self-adjoint boundary eigenvalue problems},
  author={Mennicken, Reinhard and M{\"o}ller, Manfred},
  volume={192},
  year={2003},
  publisher={Gulf Professional Publishing}
}

@book{Risken1984,
address = {Berlin, Heidelberg},
author = {Risken, Hannes},
doi = {10.1007/978-3-642-61544-3},
pages = {454},
publisher = {Springer Berlin Heidelberg},
series = {Springer Series in Synergetics},
title = {{The Fokker-Planck Equation}},
url = {http://link.springer.com/10.1007/978-3-642-61544-3},
volume = {18},
year = {1989}
}

@article{fusi1999collective,
  title={Collective behavior of networks with linear (VLSI) integrate-and-fire neurons},
  author={Fusi, Stefano and Mattia, Maurizio},
  journal={Neural Computation},
  volume={11},
  number={3},
  pages={633--652},
  year={1999},
  publisher={MIT Press One Rogers Street, Cambridge, MA 02142-1209, USA journals-info~...}
}

@article{vinci2024rosetta,
  title={Rosetta stone for the population dynamics of spiking neuron networks},
  author={Vinci, Gianni V and Mattia, Maurizio},
  journal={Physical Review E},
  volume={110},
  number={3},
  pages={034303},
  year={2024},
  publisher={APS}
}

@article{deniz2017solving,
  title={Solving the two-dimensional Fokker-Planck equation for strongly correlated neurons},
  author={Deniz, Ta{\c{s}}k{\i}n and Rotter, Stefan},
  journal={Physical Review E},
  volume={95},
  number={1},
  pages={012412},
  year={2017},
  publisher={APS}
}

@article{pietras2020low,
  title={Low-dimensional firing-rate dynamics for populations of renewal-type spiking neurons},
  author={Pietras, Bastian and Gallice, No{\'e} and Schwalger, Tilo},
  journal={Physical Review E},
  volume={102},
  number={2},
  pages={022407},
  year={2020},
  publisher={APS}
}

@article{Amit1991,
author = {Amit, Daniel J and Tsodyks, Misha},
doi = {10.1088/0954-898X/2/3/003},
journal = {Netw. Comput. Neural Syst.},
number = {3},
pages = {259--273},
title = {{Quantitative study of attractor neural network retrieving at low spike rates: I. substrate---spikes, rates and neuronal gain}},
url = {https://doi.org/10.1088/0954-898x_2_3_003},
volume = {2},
year = {1991}
}
\end{document}